\documentclass[runningheads, anonymous=true]{llncs}
\usepackage{cite}

\usepackage{amsthm}
\usepackage{amsmath,amssymb,amsfonts}
\usepackage{graphicx}
\usepackage{textcomp}
\usepackage{xcolor}
\usepackage{booktabs}
\usepackage{hyperref}

\usepackage[bottom]{footmisc}

\usepackage{siunitx}
\usepackage{lipsum}  
\usepackage{bm}  
\usepackage{algorithm}
\usepackage{tabularx}
\newcolumntype{Y}{>{\centering\arraybackslash}X}
\usepackage{algpseudocode}
\usepackage{soul}
\usepackage[toc,page]{appendix}
\usepackage{mathtools}
\usepackage{lscape}
\newtheorem{prop}[theorem]{Proposition}
\theoremstyle{definition}
\newtheorem{transformation}{Transformation}
\newtheorem*{appendixproof}{\!}

\usepackage{physics}
\usepackage{enumerate}

\usepackage{enumitem}
\setlist[itemize]{leftmargin=*}
\setlist[enumerate]{leftmargin=*}

\newcommand{\Low}[0]{\textup{\texttt{Low}}}
\newcommand{\MaxIdx}[0]{\textup{\texttt{MaxIdx}}}
\newcommand{\Comp}[0]{\textup{\texttt{Comp}}}
\newcommand{\LowComp}[0]{\textup{\texttt{LowComp}}}

\newcommand{\Plow}[0]{{\mathcal{P}_{\texttt{L}}}}
\newcommand{\Pcomp}[0]{{\mathcal{P}_{\texttt{C}}}}

\newcommand{\Reduce}[0]{\textup{\texttt{Reduce}}}
\newcommand{\scReduce}[0]{\textup{\texttt{Reduce}}}
\newcommand{\HEReduce}[0]{\textup{\texttt{HE-Reduce}}}
\newcommand{\HEReduceOptimized}[0]{\textup{\texttt{HE-Reduce-Optimized}}}

\newcommand{\mb}[1]{\mathbf{#1}}
\newcommand{\x}[0]{\mathbf{x}}
\newcommand{\y}[0]{\mathbf{y}}

\usepackage{xcolor}

\usepackage{verbatim}

\def\BibTeX{{\rm B\kern-.05em{\sc i\kern-.025em b}\kern-.08em
    T\kern-.1667em\lower.7ex\hbox{E}\kern-.125emX}}
\begin{document}

\title{An Algorithm for Persistent Homology Computation Using Homomorphic Encryption}

\author{Dominic Gold\inst{1} \and
Koray Karabina\inst{2, 3} \and
Francis C. Motta\inst{1}}
\authorrunning{D. Gold et al.}

\institute{Florida Atlantic University, Boca Raton, FL, USA \\
\email{dgold2012@fau.edu, fmotta@fau.edu}
\and 
National Research Council Canada, Ottawa, Ontario, CA
\and
University of Waterloo, Waterloo, Ontario, CA \\
\email{koray.karabina@nrc-cnrc.gc.ca}}

\maketitle

\thispagestyle{plain}
\pagestyle{plain}

\begin{abstract}
Topological Data Analysis (TDA) offers a suite of computational tools that provide quantified shape features in high dimensional data that can be used by modern statistical and predictive machine learning (ML) models. In particular, persistent homology (PH) takes in data (e.g., point clouds, images, time series) and derives compact representations of latent topological structures, known as persistence diagrams (PDs). Because PDs enjoy inherent noise tolerance, are interpretable and provide a solid basis for data analysis, and can be made compatible with the expansive set of well-established ML model architectures, PH has been widely adopted for model development including on sensitive data, such as genomic, cancer, sensor network, and financial data. Thus, TDA should be incorporated into secure end-to-end data analysis pipelines. In this paper, we take the first step to address this challenge and develop a version of the fundamental algorithm to compute PH on encrypted data using homomorphic encryption (HE). 

\end{abstract}


\section{Introduction} 
\label{section: introduction}


Topological Data Analysis (TDA) has blossomed into a suite of computational tools, built on firm mathematical theory, that generate quantified, discriminating, shape-based features of data, which can provide interpretable representations of high dimensional data and be taken in by modern statistical and predictive ML models. To apply the flagship approach, known as persistent homology (PH), data---usually in the form of point clouds or scalar-functions defined on a mesh (e.g., images, time series)---are transformed into a binary matrix that encodes the evolution of a family of simplicial complexes. From this matrix a collection of persistence diagrams (PDs) can be derived through a simple reduction algorithm. PDs provide compact representations of the number and size of geometric/topological structures in the data as multisets of planar points, and can be equipped with natural metrics. Models can then be developed either directly on PDs using, for example, hierarchical clustering or $k$-medoids in the metric space of diagrams for classification tasks, or subsequent transformations can be applied to produce topological feature vectors \cite{pis, landscapes, Perea2022, Chung2022, pmlr-v108-carriere20a, Reininghaus2015ASM} to be used with ML model architectures such as random forests, support vector machines, and neural networks. We refer to these steps as the TDA-ML pipeline, as illustrated in Figure \ref{fig:secure tda pipeline}.

\begin{figure}[t]
    \centering
    \includegraphics[width = 1.0\textwidth]{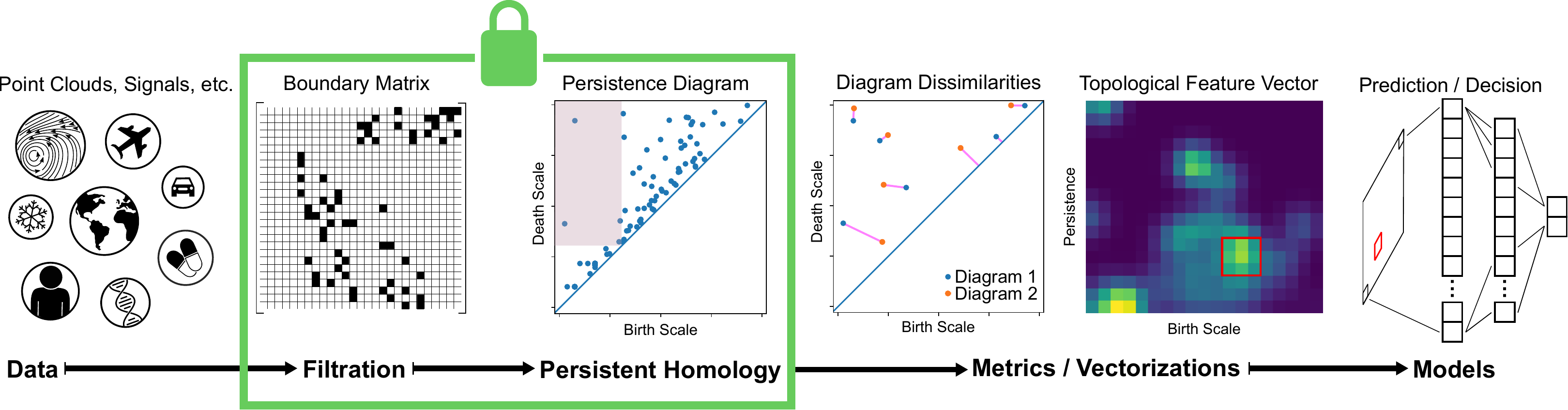}
    \caption{The TDA-ML Pipeline. Data is first transformed into a family of topological spaces encoded in a binary matrix called a boundary matrix, and then the boundary matrix is transformed into compact representations of the topological structures in the data called persistence diagrams. The distances between diagrams are computed or further transformations produce topological feature vectors. Finally, topological feature vectors are input into downstream models for a desired task. The green box indicates the contribution of this paper, securing the boundary matrix to persistence diagram step in the pipeline.}
    \label{fig:secure tda pipeline}
\end{figure}

Crucial to the use of PH in applications are the numerous stability results that establish---under a variety of assumptions about the data and the metrics placed on the data and the PDs---the (Lipschitz) continuity of the map sending data to PDs \cite{Chazal2012PersistenceSF, Cohen-SteinerEHM10, skraba2021wasserstein} and to feature vectors \cite{pis}. Due its inherent noise tolerance and suitability across domain and data types, PH has been widely adopted for model development including on sensitive data, such as genomic \cite{Shnier2019}, cancer \cite{Bukkuri2021}, sensor network \cite{Yu_2022}, and financial data \cite{GIDEA2018820}. The reader may refer to recent review articles for references to a variety of PH applications \cite{Otter2017, chazalreview}. 

As the scale of predictive models and their data demands grow, there is pressure to move to collaborative and cloud-based systems in which analysis is performed remotely and in a distributed fashion (e.g., federated learning \cite{Federated,KonecnyMRR16}). This is especially true for propriety and sensitive models that require large training data. On the other hand, a user---be it an independent data producer or agent lacking the capabilities demanded by the models---may need to keep private both their data and the decisions informed by that data. Thus, there is a growing need in industry and government for efficient, secure end-to-end data analysis pipelines that protect vital information on which sensitive decisions are made; to protect privacy, ensure compliance with personal data management regulations, and prevent hostile interference or misuse. 

Example application domains, where bridging topological data analysis and secure end-to-end algorithms will yield more efficient,
privacy-preserving, and robust applications where data analysis, data mining, statistical
inference and pattern recognition tasks are performed on private data collected from a large
number of, and potentially competing, parties include video surveillance
for law enforcement, location and energy use tracking for smart cities and autonomous
vehicles \cite{Private-Surveillance,Private-SmartCity}, financial data \cite{GIDEA2018820}, and biomedical data such as genomics \cite{Shnier2019} and cancer \cite{Bukkuri2021}, to name a few.

In order to address challenges with outsourcing sensitive data analysis, cryptographic researchers have been developing secure multiparty computing tools since the 1980s \cite{Yao86}. A good portion of the theoretical foundations of these primitives have been successfully adapted for practical applications in industry \cite{SMPC-Apps}. For example, recent innovations in homomorphic encryption (HE) have expanded the variety and complexity of the operations and algorithms that can compute on encrypted data (e.g., comparisons and conditionals \cite{encryptedcomparisons, cheon2019, FastComp2021, HE-cond}. 
Secure multiparty computing tools are nowadays interacting with privacy-preserving machine learning (ML) applications \cite{HE-LogReg, Secure-Federated, EricLogReg2020}. 
Indeed, there has been a recent surge in the development of secure ML algorithms using HE \cite{fang2021, lee2021privacypreserving, Aloufi2019, encryptedkmeans}. Thus, HE promises to expand to support complex algorithms and models that protect the privacy of both input data and model outputs. Similarly, sensitive data may be outsourced to a third party database management system (DBMS), where data owner may not fully trust DBMS but still request DBMS to perform some relational operations on the data such as sort, join, union, intersect, and difference. Specialized (symmetric key) encryption schemes allow data owners to encrypt their data, while preserving the ability of DBMS to perform such operations over the encrypted data \cite{SQLoED, CryptDB, secure_encrypted_databases, Karabina_ICDE}. 

In practice, a hybrid use of public key and symmetric encryption schemes are complementary in creating secure and trustworthy data analytical services and applications, which take encrypted data and perform both training and inference on it. Many such models have been performed this way, like logistic or ridge regression \cite{EricLogReg2020, Secure_LogReg, HE-LogReg, Ridge_Regression_Linear, Ridge_Regression_Millions}, support vector machines \cite{Fast_Homo_SVM, HE_Friendly_SVM}, random forests \cite{RF-Cryptotree, RF-DT-Inference-FPGA, RF-Fine_Grained, RF-Private}, and even neural networks \cite{NN-CryptoNets, NN-CryptoNN, NN-EMD, NN-Yongsoo}. The dual benefits of an HE framework for ML model training and inference are that while the client protects their data, the server protects their models that take in this encrypted data. In the TDA-ML pipeline (Fig. \ref{fig:secure tda pipeline}), both feature generation and model training/evaluation on those features represent critical components of the model development and deployment. Thus, each step back in the pipeline that can be realized in an HE framework relaxes the preprocessing demands on the client and strengthens the protection of the server's model. Thus securing the boundary matrix to persistence diagram step (green box in Fig. \ref{fig:secure tda pipeline}) is a critical step to allow a Server to fully protect any model that uses topological data features.

{\bf Our contributions:} We develop 
\texttt{HE-Reduce} (Algorithm~\ref{alg:secure persistence alg}) as a first-of-its-kind version of the boundary matrix reduction algorithm (\texttt{Reduce}, Algorithm~\ref{alg:standard persistence alg}), which is at the heart of PH and TDA, and which is suitable for secure computation using HE. We achieve this
by modifying the logical structure of 
Algorithm~\ref{alg:standard persistence alg} and
by developing new arithmetic circuits to replace its computational and conditional statements.
As a result, \texttt{HE-Reduce}
traces essentially the same steps as in \texttt{Reduce} but in a HE-friendly manner so that computations can be performed securely in the ciphertext space. We prove the correctness of our proposed algorithm and provide a complexity analysis. Our analysis is constructive and provides lower bounds on the implementation parameters that guarantee correctness. 
We implement our algorithms using the CKKS scheme from the OpenFHE library \cite{OpenFHE} but our techniques can be adapted for other HE schemes by implementing a compatible comparison function using BGV/BFV or TFHE schemes at a comparable cost; see \cite{Comp-BGV}. Finally, we highlight some limitations of our proposed algorithm and suggest some improvements together with some empirical evidence.

{\bf Outline:} The rest of this paper is organized as follows. Section~\ref{sec: preliminaries} establishes the mathematical and computational preliminaries of PH and HE. 
Section~\ref{sec: preliminaries} also outlines the main challenges associated with transforming \texttt{Reduce} to \texttt{HE-Reduce}.
In Section \ref{sec:secperhomologyalg}, we establish an HE-compatible version of the boundary matrix reduction algorithm, presented in Algorithm~\ref{alg:secure persistence alg}, and establish conditions guaranteeing correctness. 
Section~\ref{s: Complexity} provides a complexity analysis for Algorithm~\ref{alg:secure persistence alg} and notes on the implementation, including limitations of the proposed algorithm and potential improvements. Our plaintext implementation of Algorithm~\ref{alg:secure persistence alg} in Section~\ref{sec:empirical} simulates an implementation of \texttt{HE-Reduce} using HE, verifies the correctness of our theoretical results, and provides some positive evidence for improvements.
Our experiments showcase the propagation of errors due to relaxing algorithm parameters; see Figure~\ref{fig: approximate reduction errors}. 
We make concluding remarks in Section~\ref{s: Concluding remarks} concerning potential future research thrusts in secure TDA. In some cases, we have deferred technical proofs to the Appendix.


\section{Preliminaries}
\label{sec: preliminaries}

Our approach to adapting the PH boundary matrix reduction algorithm into a secure framework is to encrypt the input to the reduction algorithm and to allow computations to be performed on ciphertexts in such a way that the decrypted output of the algorithm is equivalent to the output of the algorithm running on the plaintext input. 
In Section~\ref{subsection: TDA}, we provide some necessary background information on PH and present the main PH boundary matrix reduction algorithm in Algorithm~\ref{alg:standard persistence alg}. In Section~\ref{subsection: HE}, we present an overview of HE and explain some of the challenges that would occur when developing a cryptographic version of Algorithm~\ref{alg:standard persistence alg}
based on HE.

We denote vectors and matrices with boldface, as in $\mb{v} \in \mathbb{R}^n$, $\mb{R} \in \mathbb{R}^{n \times n}$, and denote the $i$-th components of vectors with brackets, e.g., $\mb{v}[i]$, and columns of matrices with subscripts, $\mb{R}_i$. We denote the infinity norm of $\mb{v}$ by $\displaystyle |\mb{v}| = \lVert \mb{v} \rVert_{\infty} = \max_{i} \big| \mb{v}[i] \big|.$ We then define the following metric between any two vectors $\mb{x}, \mb{y} \in \mathbb{R}^n$ in the usual manner: $$|\mb{x} - \mb{y}| = \lVert \mb{x} - \mb{y} \rVert_{\infty} = \max_{i} \big| \mb{x}[i] - \mb{y}[i] \big|,$$ where $|\cdot|$ in the final expression is the usual absolute value of a real number. 
Furthermore, for $\mb{v} \in [0, 1]^n$, we denote $l_\mb{v} = low(\mb{v})$ as the integer-valued maximum index containing a 1 to help ease notation when appropriate.

\subsection{Persistent Homology}
\label{subsection: TDA}

\begin{figure}[t]
    \centering
    \includegraphics[width = 1.0\linewidth]{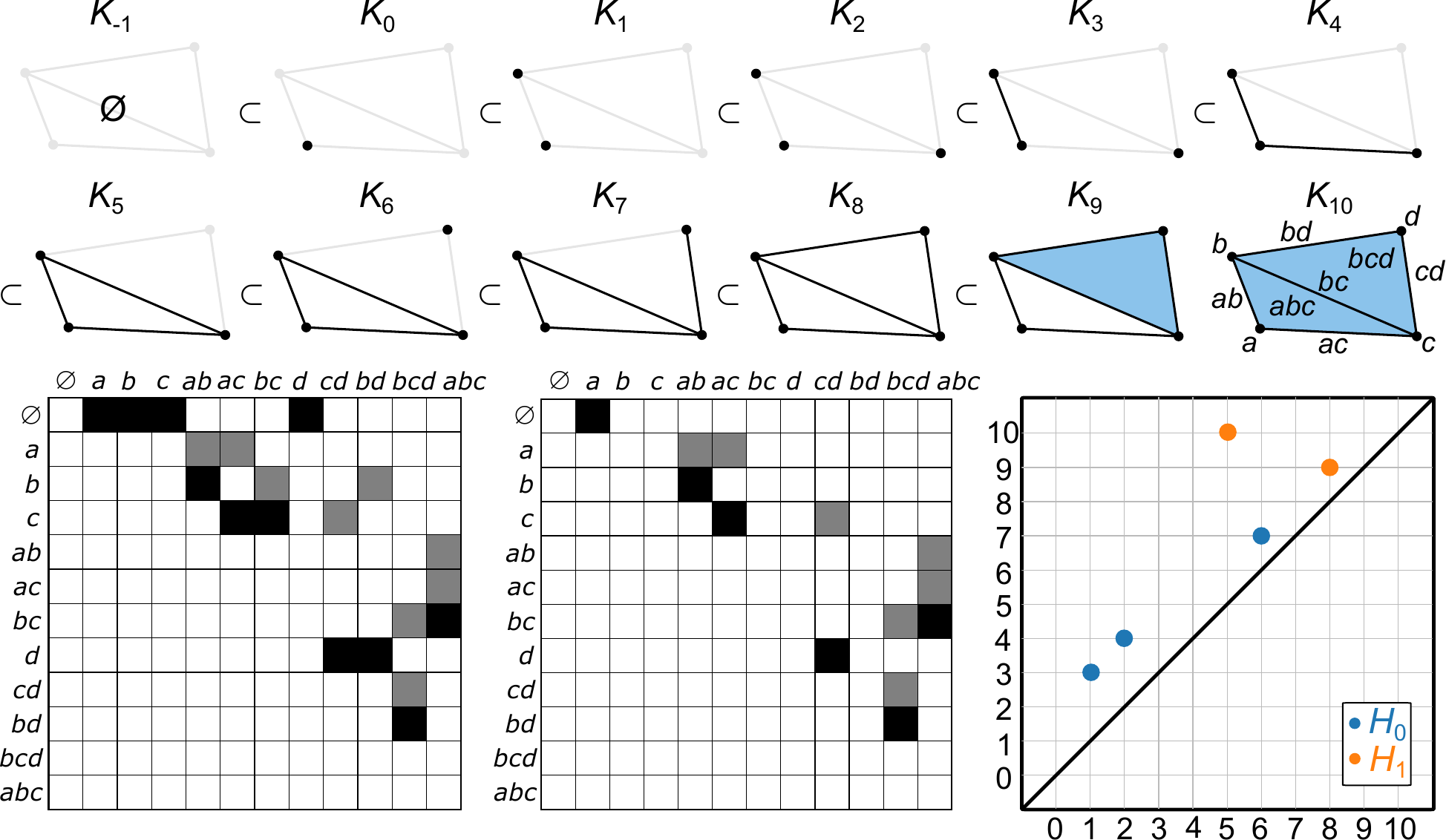}
    \caption{(Rows 1,2) Example filtration given by an ordering of the simplices in a simplicial complex that consists of 4 points, 5 edges, and 2 triangles. (Row 3) From left to right, the exact binary boundary matrix, the binary reduced boundary matrix, and the $H_0$ and $H_1$ persistence diagrams corresponding to the given filtration. White boxes in the matrices indicate 0s and shaded boxes represent 1s, with the lowest 1 in each column shaded with black.}
    \label{fig: square filtration}
\end{figure}

PH, a mathematical device from algebraic topology, provides a means of comparing data through its latent shape-based structures. This is achieved by first associating to a dataset an ordered, nested family of combinatorial objects that are equipped with well-defined notions of shape. In particular, these shape features will be representations of $k$-dimensional holes in the data. Intuitively, a $k$-dimensional hole is a vacancy left by a $(k+1)$-dimensional object whose $k$-dimensional boundary remains. In this way, PH can be regarded as a feature extraction tool which pulls from data topological/geometric features which may provide scientific insights and can be used to train discriminating or predictive models. Although there are different forms of (persistent) homology theory, we restrict our attention to simplicial homology because of its intuitive appeal and practical computability. 

\begin{definition}
  An \textit{abstract simplicial complex}, $K$, is a finite collection of finite subsets (called \textit{simplices}) such that if $\sigma \in K$, then $\tau \in K$ for all $\tau \subset \sigma$. A $k$-simplex, or a simplex of dimension $k$, is a set of size $k+1$, and the \textit{dimension} of a complex, $\text{dim}(K)$, is the maximum dimension of any of its simplices. A proper subset, $\tau \subsetneq \sigma \in K$, is called a \textit{face} of $\sigma$. If $\tau$ is a codimension-1 face of $\sigma$, i.e., $\tau \subset \sigma \in K$ and $|\tau| = |\sigma|-1$, we call $\tau$ a \textit{boundary face} of $\sigma$.  For simplicity, we will denote the $k$-simplex $\{x_0, x_1, \ldots, x_{k}\}$ by $x_0x_1 \ldots x_k$.
\end{definition}

One may regard 0-simplices (singleton sets) as points in some Euclidean space, 1-simplices (pairs) as edge segments between points, 2-simplices (sets of size 3) as filled-triangles, 3-simplices (sets of size 4) as filled tetraheda and so on, with the requirement that simplices in the geometric realization intersect only along common faces. Figure \ref{fig: square filtration} illustrates such geometric realizations of abstract simplicial complexes. For example, $K_5$ is the geometric realization of the abstract simplicial complex $\{\emptyset, a, b, c, ab, ac, bc\}$.  The empty triangle formed by the edges $ab, bc,$ and $ac$ at index 5 in Figure \ref{fig: square filtration} provides an example of a 1-dimensional hole formed by the vacancy of the missing 2-simplex, $abc$, enclosed by its three boundary edges, $ab$, $ac$, and $bc$. The holes in a simplicial complex, $K$ are collected into a group, denoted $H_1(K)$, composed of equivalence classes of collections of $1$-simplices that form cycles (e.g., $ab$, $ac$, and $bc$ in $K_5$) that could be the boundary faces of some collection of $2$-simplices, but aren't. Similarly, a collection of triangles in $K$ that enclose a void become represents of elements in $H_2(K)$. More generally, for each dimension $k$, the $k$-dimensional homology group $H_k(K)$ comprises equivalence classes of $k$-dimensional cycles that are not boundaries of a collection of $(k+1)$-dimensional simplices. $H_0(K)$ encodes the connected components of $K$.

By ordering the simplices of a simplicial complex so that no simplex appears before any of its faces, one forms a nested sequence of simplicial complexes, which we'll call a \textit{filtration}. Across this filtration one can track which simplices gave birth to homological features and which simplices kill off those homological features to determine (birth, death) pairs that track the persistence of each homological feature. For example, in Figure \ref{fig: square filtration}, $H_1$($K_4$) is trivial since $K_4$ contains no holes. This is in contrast to the complexes $K_5-K_9$ that have a non-trivial $H_1$ element represented by the boundary edges $ab, bc,$ and $ac$ that was born with the introduction of $bc$ at index 5. In $K_8$ there appears another hole with the introduction of the edge $bd$, which then disappears in $K_9$ when the triangle $bcd$ fills the cycle formed by $bc$, $bd$, $cd$.

In practice one usually defines a complex, $K$, from a dataset and computes a filtration from a real-valued function $f: K \rightarrow \mathbb{R}$ that satisfies $f(\tau) \leq f(\sigma)$ if $\tau \subseteq \sigma \in K$. $f$ encodes the `scales' at which each simplex appears in the filtration gotten by ordering simplices according to their scales and sorting ties arbitrarily while ensuring each simplex never appears before its faces. A multitude of methods have been proposed to derive such filtrations \cite{PH_roadmap}, both from point cloud data (e.g., Vietoris-Rips filtration \cite{ZOMORODIAN2010263}, alpha filtration \cite{Edelsbrunner1995}) and related filtrations for functions on a cubical mesh \cite{Wagner2012}. However determined, the structures in the filtration can be encoded in a square, binary matrix ${\mb \Delta}({\mb K})$ called a \textit{boundary matrix}, whose rows and columns are indexed by the simplices in $K$, ordered $\sigma_1, \ldots, \sigma_n$ so that $i < j$ if $f(\sigma_i) < f(\sigma_j)$ or if $\sigma_i \subset \sigma_j$. The entries of the boundary matrix are
\[{\mb \Delta}_{i,j} = \begin{cases} 1,& \text{ if } \sigma_i \text{ is a boundary face of } \sigma_j \\ 0, & \text{otherwise} \end{cases}.\]
 Thus, ${\mb \Delta}$ encodes the order in which simplices appear in the filtration and the relationship between each simplex and its boundary simplices. We let the first row and column correspond to the empty simplex, $\emptyset$, so that the vertices have boundary equal to $\emptyset$. Thus, vertices are encoded by a column [1,0,\ldots,0], while ${\mb \Delta}_{0}$ is then necessarily a zero column, which could be omitted. The scales, $f(\sigma_i)$, at which each simplex is added to the complex may be regarded as a real-valued vector in $\mathbb{R}^n$ and can be held separately from the combinatorial information encoded in the boundary matrix.
\begin{algorithm}[t]
\caption{\texttt{Reduce}(${\mb \Delta}$)}
\label{alg:standard persistence alg}
\flushleft\textbf{Input:} A boundary matrix ${\mb \Delta} = [\mb{\Delta}_0 \: | \: \mb{\Delta}_1 \: | \: ... \: | \: \mb{\Delta}_{n-1}] \in \mathbb{Z}_2^{n \times n}$ \\
\textbf{Output:} A reduced matrix $\mb{R} \in \mathbb{Z}_2^{n \times n}$
\begin{algorithmic}[1]
\State $\mb{R} \gets {\mb \Delta}$
\For {$j \gets 0$ \textbf{to} $n-1$}
    \While {$\textbf{ exists } j_0 < j \textbf{ with } low(\mb{R}_{j_0}) = low(\mb{R}_{j})$}
        \State $\mb{R}_{j} \gets (\mb{R}_{j} + \mb{R}_{j_0}) \mod 2$
    \EndWhile
\EndFor
\State \Return $R$
\end{algorithmic}
\end{algorithm}

It is shown in \cite{persalg, edelsharer} that calculation of the persistence pairs can be achieved through a straightforward algorithm (Algorithm \ref{alg:standard persistence alg}) that brings a boundary matrix into a reduced form. The critical operation needed to transform a filtered simplicial complex ${\mb K}$---given by the monotonic filtration function $f: K \rightarrow \mathbb{R}$ and encoded in a boundary matrix ${\mb \Delta}$---into its PDs is the function 
$$
low(\mb{v}) = \max (\{i \; | \; (\mb{v}[i] = 1\}),
$$
which returns the largest index among those coordinates of the binary vector $\mb{v}$ that are equal to 1. Progressing from $j=1$ to $n$ (i.e., in the order of the simplices given by the monotonic function $f$), each column ${\mb \Delta}_j$ is replaced with the mod-2 sum ${\mb \Delta}_i + {\mb \Delta}_j$, whenever $\text{low}({\mb \Delta}_i) = \text{low}({\mb \Delta}_j)$ and $i<j$, until the lowest 1 in column $j$ is distinct from all lowest 1s in the preceding columns. The lowest 1s in the reduced boundary matrix then specify the indices of the pair of simplices at which each PH class of the corresponding dimension is born and dies. More precisely, let ${\mb R} = \texttt{Reduce}({\mb \Delta})$ be the reduction of the boundary matrix ${\mb \Delta}$ after applying Algorithm \ref{alg:standard persistence alg}. Then $(f(\sigma_i), f(\sigma_j))$ is a (finite persistence) point in the $k$-dimensional PD $\text{dgm}_k({\mb K})$ if and only if  $\sigma_i$ is a simplex of dimension $k$ and $i = \text{low}({\mb R_j})$. In other words, a $k$-dimensional homology class was born with the introduction of the simplex $\sigma_{i} = \sigma_{\text{low}({\mb R}_j)}$ and died when $\sigma_j$ was added to the filtration. 

In Figure \ref{fig: square filtration} we illustrate the original boundary matrix, its reduced form after applying Algorithm \ref{alg:standard persistence alg}, and $H_0$ and $H_1$ PDs associated to the given filtration. In the reduced matrix, columns $b$ and $c$ consist of all zeros, since their appearance created homology ($H_0$) classes\footnote{The first vertex $a$ is a special case, and technically kills off the (-1)-dimensional (reduced) homology class at index -1.}. The connected components represented by vertices $b$ and $c$ are then killed by the introduction of $ab$ and $ac$ respectively, since these edges merge the connected component into the component represented by $a$, which was born earlier. This is encoded in the reduced boundary matrix by the low 1s at indices ($b$, $ab$) and ($c$, $ac$) respectively. The edge $bc$ likewise gives birth to an $H_1$ class, that is later killed off by the introduction of the triangle $abc$. This is why, in $\texttt{Reduce}({\mb \Delta})$, column $bc$ consists of all zeros and the low 1 in $abc$ is in row $bc$. 
    
The low 1s in the reduced matrix encode the birth-death simplex pairs appearing in the PDs of the filtration. Here we take the scale of each simplex to be the index of the complex in which it first appears so that the low 1 at ($b$,$ab$) is sent to the point (1,3) in the $H_0$ diagram. Similarly, ($bd$, $bcd$) maps to (8,9) and ($bc$, $abc$) maps to (5,10) in the $H_1$ PD, $\text{dgm}_1({\mb K}).$ If the scales of each simplex were determined instead by some geometric information in the data (e.g., using pairwise distances between points as is the case for the Vietoris-Rips filtration), the positions of the points in the PDs would capture these scales, rather than merely the indices.

\subsection{Homomorphic Encryption}
\label{subsection: HE}
Let $\mathcal{M}$ be a message (plaintext) space and 
$\mathcal{C}$ be a ciphertext space. We assume that $\mathcal{M}$ and $\mathcal{C}$ are commutative rings with their respective identity elements, and addition and multiplication operations, denoted
$(\mathcal{M}, 1_{\mathcal{M}}, +, \times)$ and $(\mathcal{C}, 1_{\mathcal{C}}, \oplus, \otimes)$. When the underlying ring is clear from the context, we simply denote the identity element by $1$ and so by abuse of notation the scalar space of the ring consists of elements $s = \sum_{i=1}^{s}{1} \in \mathbb{Z}$.
For a given parameter set \texttt{params},
an HE scheme consists of algorithms as described in the following:
\begin{itemize}
    \item \texttt{KeyGen}(\texttt{params}): Takes \texttt{params} as input, and outputs a public key and secret pair (\texttt{pk}, \texttt{sk}), and an evaluation key \texttt{evk}.
    \item $\texttt{Enc}_{\texttt{pk}}(\mb{m})$: Takes  
    a plaintext message $\mb{m}\in \mathcal{M}$ and the public key \texttt{pk} as input, and outputs a ciphertext $c\in \mathcal{C}$.
    \item $\texttt{Dec}_{\texttt{sk}}(c)$: Takes  
    a ciphertext $c\in \mathcal{C}$ and the public key \texttt{sk} as input, and outputs a plaintext message $\mb{m}\in \mathcal{M}$.
    \item $\texttt{Add}_{\texttt{evk}}(c_1, c_2)$: Takes a pair of ciphertexts $(c_1, c_2)$,  $c_i \in \mathcal{C}$ and the evaluation key \texttt{evk} as input, an outputs a ciphertext $c_{\texttt{add}} \in \mathcal{C}$.
    \item $\texttt{Mult}_{\texttt{evk}}(c_1, c_2)$: Takes a pair of ciphertexts $(c_1, c_2)$,  $c_i \in \mathcal{C}$ and the evaluation key \texttt{evk} as input, an outputs a ciphertext $c_{\texttt{mult}} \in \mathcal{C}$.
    \item $\texttt{Eval}_{\texttt{evk}}(f; c_1, ..., c_k)$: Takes 
     an arithmetic circuit $f: \mathcal{M}^k \rightarrow \mathcal{M}$, ciphertexts $c_i\in\mathcal{C}$, and the evaluation key \texttt{evk} as input, and outputs a ciphertext $c_{\texttt{eval}} \in \mathcal{C}$.
\end{itemize}

Here, \texttt{params} generally consists of a security parameter $\lambda$ and a multiplicative depth parameter $L$. The security parameter
$\lambda$ says that complexity of the best attack to break the security of the HE scheme is $\mathcal{O}$. The depth parameter $L$ guarantees that the HE scheme can evaluate circuits of maximum multiplicative depth $L$.
 We frequently refer to multiplicative depth and computational complexity of circuits in our analysis and they are defined as follows.
\begin{definition}
Let $f$ be an arithmetic circuit. \textit{Multiplicative depth}, or simply \textit{depth}, of $f$ is the maximum number of sequential multiplications required to compute $f$. \textit{Computational complexity}, or simply \textit{complexity}, of $f$ is the number of multiplication and addition operations required to compute $f$.
\end{definition}
For example, $f(m_1, m_2, m_3, ..., m_n) = \sum_{i=1}^{n}{m_i^{2^i}}$ is a depth-$n$ multiplicative circuit, where $m_i^{2^i}$ can be computed after $i$ successive multiplications (squarings). A naive way to compute $f$ would require $n(n+1)/2$ multiplications and $(n-1)$ additions and so we can say that $f$ has computational complexity $\mathcal{O}(n^2)$.

A basic correctness requirement\footnote{\label{neg_prob}The correctness and homomorphic features of HE may be violated with negligible probability.} for an HE scheme is that the decryption operation is the inverse of the encryption operation, that is 
\[\texttt{Dec}_{\texttt{sk}}(\texttt{Enc}_{\texttt{pk}}(m)) = m\] 
for all $m\in \mathcal{M}$.
The homomorphic feature\footref{neg_prob} of an HE scheme requires
\[\texttt{Dec}_{\texttt{sk}}(\texttt{Eval}_{\texttt{evk}}(f; c_1, ..., c_k)) = f(m_1, ..., m_k)\]
for all $c_i\in \mathcal{C}$ such that  $c_i= \texttt{Enc}_{\texttt{pk}}(m_i)$.
In other words, HE allows one to evaluate polynomials
on encrypted data such that the decryption of the result is exactly the same as the value of that polynomial evaluated on plaintext messages. We should note that we presented here a limited overview of HE schemes so that our paper is self-contained. HE schemes are much more involved (e.g., ~consisting of other algorithms such as scaling, relinearization, bootstrapping, etc.) and their implementations require a great deal of detail (e.g., ~encoding and decoding algorithms so that the plaintext messages can be mapped into the  message space of the scheme, batching operations, etc.). Moreover, most of these details depend on the choice of the HE scheme. For a survey of HE schemes and existing libraries, we refer the reader to \cite{HE_survey2018, OpenFHE}.

Some of the challenges of using HE in practice are:
\begin{itemize}
    \item Increasing the depth of the arithmetic circuit significantly increases the complexity of the circuit's encrypted evaluation. Practical HE schemes can handle arithmetic circuits with relatively low depth. For example, \cite{EricLogReg2020} reports and compares some results for homomorphic evaluation of circuits up to depth $30$. Bootstrapping is a viable option to reset the level of a ciphertext right before maximum tolerance is reached.
    \item Algorithms in general require evaluation of functions that are not necessarily polynomials and approximation of functions through low-depth circuits is a challenge. Similarly, algorithms involve conditional statements and evaluating these statements while running an algorithm on ciphertext variables requires different ways of handling conditionals. As an example, given $m_1, m_2\in\mathbb{Z_p}$ for some prime $p$, the conditional statement that returns $m_1+m_2$ if $m_1=m_2$; and that returns $m_1$ if $m_1\ne m_2$ can be implemented over ciphertexts as 
    $$
    (\texttt{Eval}_{\texttt{evk}}(f; c_1,c_2)\otimes c_1)
    \oplus ((1-\texttt{Eval}_{\texttt{evk}}(f; c_1,c_2))\otimes(c_1\oplus c_2)),
    $$
    where $c_i = \texttt{Enc}_{\texttt{pk}}(m_i)$, and $f(m_1, m_2) = (m_1-m_2)^{p-1}$ can be implemented as an arithmetic circuit of depth $\mathcal{O}(\log_2{p})$ using a square-and-multiply type exponentiation algorithm.
\end{itemize}

Our objective is to adapt 
Algorithm~\ref{alg:standard persistence alg}
so that secure boundary matrix reduction operation can be performed based on encrypted boundary matrices
using HE. In the light of our discussion above, there are three main challenges to address:
\begin{enumerate}
    \item Develop an arithmetic circuit for {\it encrypted low} computations so that given 
    a pair of ciphertexts $c_1$ and $c_2$ 
    (representing the encryption of column vectors $\mb{v}_1$ and $\mb{v}_2$), $low(\mb{v}_1) = low(\mb{v}_2)$ can be verified; see line 3 in Algorithm~\ref{alg:standard persistence alg}.
    \item Develop an arithmetic circuit so that the conditional modular addition operation (line 4 in Algorithm~\ref{alg:standard persistence alg}) can be performed in the ciphertext space.
    \item Modify the logical structure of Algorithm~\ref{alg:standard persistence alg} so that all of the modular vector additions in lines 2-4 in Algorithm~\ref{alg:standard persistence alg} are correctly executed in the ciphertext space, until $low(\mb{R}_{j_0})\ne low(\mb{R}_j)$ for all $j_0<j$, and for all $j = 0,...,(n-1)$.
\end{enumerate}

\section{HE-Compatible Matrix Reduction}
\label{sec:secperhomologyalg}
\subsection{$\Low$: HE-compatible computation of $low$}
\label{s: Secure low}

The first obstacle to realizing an HE-compatible \texttt{Reduce} algorithm is computing the largest index of any 1 in an $n$-dimensional binary vector $\mb{v} \in \{0, 1\}^n$, called $low(\mb{v})$ (see Section \ref{subsection: TDA}). For reasons that will become clear, it will be necessary for us to extend the usual definition of $low$---as defined in Section \ref{sec: preliminaries}---to the $n$-dimensional 0-vector; we assign $low(\mb{0}) = n-1$.  By construction, a non-zero column in a boundary matrix of a valid filtration can never have a $low$ of $n-1$ before or during reduction by Algorithm \ref{alg:standard persistence alg}. \footnote{If it did, that would imply the simplex that appeared latest is the boundary of a simplex that appeared earlier, which violates the condition that each step in the filtration gives a valid complex.}

In \cite{cheon2019}, the authors introduce a method of locating the index of the maximum value of a vector ($maxidx$) of distinct numbers
using HE. We adapt this method to obtain an approximation of the $low$ value of a binary vector. First, in \textsc{Lemma}~\ref{lemma:low_maxidx_exact}, we establish the correctness of our reimagining of the exact $low$ function obtained by monotonically scaling vector coordinates with respect to their index while ensuring all coordinates remain distinct and guaranteeing the $low$ corresponds to the new largest coordinate.

\begin{transformation}
\label{trans:S}
    For $\mb{v} \in \mathbb{R}^n$, let $S\left(\mb{v}\right) := \left[\mb{v}[i] + \frac{i}{n}\right]_{i=0}^{n-1}$
\end{transformation}

\begin{definition}
Let $\mathcal{D}^n = \{\mb{v} \in \mathbb{R}^n \; | \; \mb{v}[i] \neq \mb{v}[j], 0 \leq i \neq j < n\}$ be the collection of $n$-dimensional vectors with distinct coordinates. For a vector $\mb{v} \in \mathcal{D}^n$, define
$maxidx: \mathcal{D}^n  \rightarrow \mathbb{Z}$ by
$
maxidx(\mb{v}) = k
$
if $\mb{v}[k] > \mb{v}[j]$ for all $j$ different from $k$.
\end{definition}




\begin{lemma}
\label{lemma:low_maxidx_exact}
For any binary vector $\mb{v} \in \{0, 1\}^n$, $$low(\mb{v}) = maxidx(S(\mb{v}))$$
\end{lemma}
\begin{proof}
See Appendix~\ref{s: Low correctness proofs}.
\end{proof}

How does our argument about $maxidx$ approximating $low$ hold in our ``approximate arithmetic'' setting? The following generalization of \textsc{Lemma}~\ref{lemma:low_maxidx_exact} states that as long as our \textit{approximate} binary vector $\mb{v'} \in \mathbb{R}^n$ isn't too far from an underlying, true binary vector $\mb{v} \in \{ 0, 1 \}^n$, then we may continue to extract $low(\mb{v})$ using $maxidx(\mb{v'})$.

\begin{lemma}
\label{lemma:low_maxidx_approx}
Let $\mb{v}\in \{0, 1\}^n$ and $\mb{v'}\in \mathbb{R}^n$ be given such that $\left| \mb{v'} - \mb{v} \right| < \frac{1}{2n}$. Then 
\[low(\mb{v}) =  maxidx(S(\mb{v'})).\] 
\end{lemma}
\begin{proof}
See Appendix~\ref{s: Low correctness proofs}.
\end{proof}

\begin{remark} The proximity between $\mb{v}$ and $\mb{v'}$ cannot be relaxed for the above choice of Transformation~\ref{trans:S}, since it is possible to construct vectors, $\mb{v}$ and $\mb{v'}$ such that  $|\mb{v} - \mb{v'}| = \frac{1}{2n}+ c$, with $0 = low(\mb{v}) \neq maxidx(S(\mb{v}')) = n-1$ for any $c > 0$.
\end{remark}

Using this construction, it is then natural to apply the $\MaxIdx$ function presented in \cite{cheon2019} (Algorithm \ref{alg:MaxIdx}), to develop the \texttt{Low} function (Algorithm \ref{alg:low}). This $\Low$ function will estimate $low$ with arbitrary accuracy, for real vectors that well-approximate binary vectors. 



\begin{algorithm}[t]
\caption{\texttt{MaxIdx}($\mb{v}; d, d', m, t$) from \cite{cheon2019}}
\label{alg:MaxIdx}
\flushleft \textbf{Input:} A vector $\mb{v} \in [\frac{1}{2}, \frac{3}{2})^n \cap \mathcal{D}^n$; $d, d', m, t \in \mathbb{N}$ \\
\textbf{Output:} A vector $\mb{b} \in [0,1]^n$ such that $\mb{b}[k] \approx 1$, if $maxidx(\mb{v}) = k$, otherwise $\mb{b}[i] \approx 0$.\\
\textbf{Depth:}  $d' + 1 + t(d + \log m + 2)$\\
\textbf{Complexity:} $O(n + d' + t(d + n \log m))$
\begin{algorithmic}[1]
\State $I \gets \texttt{Inv}(\sum_{j = 0}^{n-1} \mb{v}[j]/n; d')$
\For {$j \gets 0$ \textbf{to} $n-2$}
    \State $\mb{b}[j] \gets \mb{v}[j]/n \cdot I$
\EndFor
\State $\mb{b}[n-1] \gets 1 - \sum_{j = 0}^{n-2} \mb{b}[j]$
\For {$i \gets 1 \textbf{ to } t$}
    \State $I \gets \texttt{Inv}(\sum_{j = 0}^{n-1} \mb{b}[j]^m; d)$
    \For {$j \gets 0 \textbf{ to } n-2$}
        \State $\mb{b}[j] \gets \mb{b}[j]^m \cdot I$
    \EndFor
    \State $\mb{b}[n-1] \gets 1 - \sum_{j = 0}^{n-2} \mb{b}[j]$
\EndFor
\State \Return $\mb{b}$
\end{algorithmic}
\end{algorithm}

\texttt{MaxIdx} takes a vector $\mb{v} \in [1/2, 3/2)^n$ and returns a vector $\mb{b}$ with $\mb{b}[k] \approx 1$ if $maxidx(\mb{v}) = k$ and $\mb{b}[j] \approx 0$ for $j \neq k$. The component-wise accuracy in approximating the coordinates of the true maximum value indicator vector ($\mb{b}$ with $\mb{b}[maxidx(\mb{v})] = 1$ and 0 elsewhere) is controlled by a tuple of parameters $\mathcal{P}_{\texttt{L}} = (d, d', m, t)$. In \cite{cheon2019}, the authors show that the error in each coordinate is bounded by $2^{-\alpha}$, for $\alpha > 0$ which can be made arbitrarily large with sufficiently large choices of $d, d', m$, and $t$.

To attain the actual index containing the maximum value of $\mb{v}$, as opposed to the maximum index indicator vector, $\mb{b}$, we compute the dot product between $\mb{b}$ and $[0, 1, ..., n-1]$. This is the approach we adopt in the $\texttt{Low}$ function given in Algorithm \ref{alg:low}. Since the $\MaxIdx$ algorithm requires the input vector to be in the interval $[\frac{1}{2}, \frac{3}{2})^n$, and our inputs $S(\mb{v}')$ will be in the interval $[0, 2)^n$, we apply a linear transformation that preserves the $maxidx$ of its input.

\begin{transformation}
\label{trans:T low}
$T_{\texttt{L}}\left(\mb{v}\right) := \left[ \frac{\mb{v}[i] + 1}{2} \right]_{i=0}^{n-1}$ 
\end{transformation}

The error in $\texttt{MaxIdx}$ propagates through the $\texttt{Low}$ algorithm in the following manner:
\begin{theorem}
\label{theorem: lowerror}
Let $\alpha > 0$ and fix parameters $d, d', m, t$ for the $\MaxIdx$ algorithm so that 
\[|\MaxIdx(\mb{x}; d,d',m,t) - \mb{e}_{maxidx(\mb{x})}| < 2^{-\alpha},\] 
for all $\mb{x} \in [\frac{1}{2}, \frac{3}{2})^n$.
Further assume $\mb{v'} \in [0, 1]^n$ and $\mb{v} \in \{ 0, 1 \}^n$ are such that $|\mb{v'} - \mb{v}| < \frac{1}{2n}$. Then $$\big| \Low(\mb{v'}; d, d', m, t) - low(\mb{v}) \big| < \frac{3}{2}(n)(n-1)  2^{-\alpha}.$$
\end{theorem}

\begin{proof}
The result follows from \textsc{Lemmas} \ref{lemma:lowerrorexact} and \ref{lemma:lowerrorapprox} in Appendix~\ref{s: Low correctness proofs} and the triangle inequality.
\end{proof}

In the next section we establish choices of parameters ensuring a specified level of accuracy of the approximating \texttt{Low} function. 


As a final remark, we note that the dependence of \texttt{Low}'s error on $n^2$ is a consequence of extracting the $low$ of a vector using a dot product between the vector of indices, $[0,\ldots,n-1]$, and the max-index-indicator vector. This may be unavoidable when using the current implementation of the \texttt{MaxIdx} function, although it is conceivable that a fundamentally different approach to computing $\Low$ may yield a better error growth with the size of the boundary matrix.


\begin{algorithm}[t]
\caption{\texttt{Low}($\mb{v}'; d, d', m, t$)} 
\label{alg:low}
\flushleft \textbf{Input:} A vector $\mb{v'} \in [0, 1]^n$; $d, d', m, t \in \mathbb{N}$ \Comment{$\mb{v'} \approx \mb{v} \in \{0, 1\}^n$} \\
\textbf{Output:} A real number $r$ \Comment{$r \approx low(\mb{v})$}\\
\textbf{Depth:} $d' + 2 + t(d + \log(m) + 2)$ \\
\textbf{Complexity:} $O(n + d' + t(d + n \log m))$ 
\begin{algorithmic}[1]
\State $\mb{v'} \gets S(\mb{v'})$ \Comment{\textsc{Lemma}~\ref{lemma:low_maxidx_approx}}
\State $\mb{v'} \gets T_{\texttt{L}}(\mb{v'})$ \Comment{Maps $[0, 2)^n$ to $[1/2, 3/2)^n$}
\State $\mb{b} \gets \texttt{MaxIdx}(\mb{v'}; d, d', m, t)$ \Comment{Algorithm \ref{alg:MaxIdx}}
\State $r \gets \mb{b} \cdot [0, 1, ..., n-1]$ \Comment{Extract $low$ estimate}
\State \Return $r$
\end{algorithmic}
\end{algorithm}  
\subsection{Parameters for $\Low$}
\label{subsection: low parms}

Having established an approximation of the $low$ function that is amenable to an HE framework, we next establish the prerequisite results needed to inform the choices of \texttt{Low}'s parameters that will guarantee correctness. There are two results we create in order to help ease the proof of the theorem at the end of this section. The first of these is to establish a lower bound on this ratio over for all binary vector inputs to \texttt{Low}, as this value will directly affect the choice of parameters for the $\MaxIdx$ and, subsequently, the $\Low$ functions. 

Let us borrow Theorem 5 from \cite{cheon2019}, which gives the parameter choices $(d, d', m, t)$ to achieve any desired non-zero error 
\[|\MaxIdx(\mb{v}; d, d', m, t) - \mb{e}_{maxidx(\mb{v})}| < 2^{-\alpha}.\]
\begin{theorem}[Theorem 5 in \cite{cheon2019}]
\label{MaxIdx Parameter Choices}
    Let $\mb{v} \in [\frac{1}{2}, \frac{3}{2})^n$ be a vector with $n$ distinct entries. Define $c$ to be the ratio of the maximum value over the second maximum value such that $c \in (1, 3)$. If 
    \begin{align*}
    t &\geq \frac{1}{\log(m)}[\log(\alpha + \log(n) + 1) - \log \log (c)] \\
    \min(d, d') &\geq \log(\alpha + t + 2) + (m-1)\log(n) - 1
    \end{align*}
then the error (component-wise) of the $\MaxIdx(\mb{v}; d, d', m, t)$ algorithm compared to $\mb{e}_{maxidx(\mb{v})}$ is bounded by $2^{-\alpha}$.
\end{theorem}

Of great importance to us is a lower bound on $c$, the ratio of the largest to the second largest coordinate values in the input to $\MaxIdx$'s parameters. As $c$ approaches 1, $\MaxIdx$ and $\Low$'s parameters $d, d'$, and $t$ grow without limit. For this reason, we aim to obtain a larger lower bound on $c$ across all possible (approximate binary) input vectors. We re-write the bound $| \mb{v} - \mb{v}' | < \frac{1}{2n}$ as $|\mb{v} - \mb{v}'| \leq \frac{\varepsilon}{2n}$ where $\varepsilon \in [0, 1)$ to fine-tune parameter $c$.


We compute that a lower bound on $c$ is given by $c \geq 1 + \frac{2-2\varepsilon}{6n-4+\varepsilon}$ in \textsc{Lemma}~\ref{prop:c bound for approx} in Appendix~\ref{s: Low correctness proofs}. Importantly, if $\varepsilon = 1$ (and so assume $\mb{v}'$ is approximately binary only within the bound $1/2n$ needed for \textsc{Lemma}~\ref{lemma:low_maxidx_approx} to compute $low$ via $maxidx$) then the ratio of the first to the second largest coordinates of the transformed $\mb{v}'$ can be arbitrarily close to 1. As a consequence, there will no longer exist a choice of finite parameters in the \texttt{Low} algorithm that guarantees correctness over all possible approximately-binary vectors $\mb{v}'$.
On the other hand, as $\varepsilon$ gets closer to 0, the lower bound on $c$ increases away from 1, which will allow $\texttt{Low}$ to be computed more efficiently. Thus there will be a trade-off between the computational cost of maintaining $\mb{v}'$ sufficiently close to binary throughout the boundary matrix reduction, and estimating $low$ efficiently.

The variable $\alpha$ specifies the desired level of accuracy of $\MaxIdx$ (to $2^{-\alpha}$), and informs the minimum parameters needed to attain said accuracy. \textsc{Lemma}~\ref{lemma: alpha to delta} recasts the accuracy parameter of \Low \; to an arbitrary $\delta > 0$. With this, we can specify the choice of parameters needed to approximate $low(\mb{v})$ using $\Low(\mb{v'}; d, d', m, t)$ to arbitrary accuracy. 

\begin{theorem}
\label{theorem: low parms}
Assume $\mb{v} \in \{ 0, 1 \}^n$ and $\mb{v'} \in [0, 1]^n$ are such that $\left| \mb{v} - \mb{v'} \right| \leq \frac{\varepsilon}{2n}$, for some $0 \leq \varepsilon < 1$. Choose the parameters $d, d', m$, and $t$ for the $\MaxIdx$ function, along with a pre-determined $\delta>0$, such that
\begin{align*}
    \alpha &>   \log(3) + 2\log(n) - \log(\delta)-1   \\
    t &\geq \frac{\log\Big(\alpha + 1 + \log(n) \Big) - \log \log \Big( 1 + \frac{2-2\varepsilon}{6n-4+\varepsilon} \Big)}{\log m} \\
    \min (d, d') &\geq  \log (\alpha + t + 2) + (m-1)\log (n) - 1 
\end{align*}
Then $\Low(\mb{v}'; d, d', m, t)$ has $\delta$-error. That is, 
\[\left|\Low(\mb{v'}; d, d', m, t) - low(\mb{v}) \right| < \delta.\]
\end{theorem}
\begin{proof}
See Appendix~\ref{s: Low parameters proofs}.
\end{proof}


As these parameters are now well-established for the $\Low$ function, we now refer to this tuple of parameters $(d_{\texttt{L}}, d'_{\texttt{L}}, m_{\texttt{L}}, t_{\texttt{L}})$ as $\mathcal{P}_{\texttt{L}}$ to avoid confusion with the upcoming \texttt{Comp} function which will have a similar parameter naming convention. Furthermore, when $\mathcal{P}_{\texttt{L}}$ is clear from context, define $$\texttt{L}_{\mb{v}} \coloneqq \Low(\mb{v}; \mathcal{P}_{\texttt{L}})$$ for ease of notation in the upcoming sections.

\subsection{$\LowComp$: HE-compatible Equality Check} 
\label{subsection: LowComp}

\textsc{Theorem}~\ref{theorem: low parms} approximates $low(\mb{x})$ and $low(\mb{y})$ via  
$\Low(\mb{x^\prime}; \mathcal{P}_{\texttt{L}})$ and $\Low(\mb{y^\prime}; \mathcal{P}_{\texttt{L}})$. One of the remaining challenges is to characterize the
equality check $low(\mb{x}) = low(\mb{y})$ using 
$\Low(\mb{x^\prime}; \mathcal{P}_{\texttt{L}})$ and $\Low(\mb{y^\prime}; \mathcal{P}_{\texttt{L}})$. The second challenge is to rewrite \eqref{eq: original modular vector update} for $\mb{z}^\prime$ so that it can be computed by avoiding the if statement and the mod 2 addition.

Suppose that $\mb{x}^\prime$ and $\mb{y}^\prime$
are two real valued vectors that are approximations of the binary vectors $\mb{x}$ and $\mb{y}$, respectively. 
We must now determine a method that takes $\mb{x}^\prime$ and $\mb{y}^\prime$ as input, and outputs $\mb{z}^\prime$ such that $\mb{z}^\prime$ approximates the binary vector
\begin{align}
\label{eq: original modular vector update}
\mb{z} = 
\begin{cases}
\mb{x} + \mb{y} \mod 2 & \text{if } low(\mb{x}) = low(\mb{y})\\
\mb{x}  & \text{if } low(\mb{x}) \not= low(\mb{y})
\end{cases}
\end{align} 
\noindent In Section~\ref{subsection: mod 2 addition}, we show that 
$\mb{z}$ in \eqref{eq: original modular vector update} can be approximated by
\begin{align}
\label{eq:approximating_column_update}
\mb{z}^\prime = \Omega(\mb{x^\prime}-\mb{y^\prime})^2 + (1-\Omega)\mb{x^\prime},
\end{align}
where the predicate $\Omega$ takes $\Low(\mb{x^\prime}; \mathcal{P}_{\texttt{L}})$ and $\Low(\mb{y^\prime}; \mathcal{P}_{\texttt{L}})$
as input, and approximates the boolean value $low(\x)==low(\y)$. We establish the theory to calculate $\Omega$ in this section.

\begin{lemma} 
\label{lem:approxlow}
Let $\mb{x}, \mb{y} \in \{0, 1\}^n$ and $\mb{x}', \mb{y}' \in [0, 1]^n$ and assume that $\mathcal{P}_{\texttt{L}}$ is chosen such that $\lvert \Low(\mb{x}'; \mathcal{P}_{\texttt{L}}) - low(\mb{x}) \rvert < \delta$ and $\lvert \Low{(\mb{y}'; \mathcal{P}_{\texttt{L}})} - low(\mb{y}) \rvert < \delta$ for some $0 < \delta < \frac{1}{4}$. Let $\phi$ be any value in the interval $(2\delta, 1-2\delta)$. Then $$\lvert \Low{(\mb{x'}; \mathcal{P}_{\texttt{L}})} - \Low{(\mb{y'}; \mathcal{P}_{\texttt{L}})} \rvert \leq \phi \text{ iff } low(\mb{x}) = low(\mb{y})$$
\end{lemma}

\begin{proof}
Suppose that $\lvert \Low{(\mb{x}'; \mathcal{P}_{\texttt{L}})} - \Low{(\mb{y}'; \mathcal{P}_{\texttt{L}})} \rvert > \phi$. Then
\begin{align*}
    \phi &< | \Low{(\mb{x}'; \mathcal{P}_{\texttt{L}})} - \Low{(\mb{y}'; \mathcal{P}_{\texttt{L}})} | \\
    & \leq | \Low{(\mb{x}'; \mathcal{P}_{\texttt{L}})} - low(\mb{x}) | + | \Low{(\mb{y}'; \mathcal{P}_{\texttt{L}})} - low(\mb{y}) | + | low(\mb{x}) - low(\mb{y}) | \\
    &< 2\delta + | low(\mb{x}) - low(\mb{y}) |.
\end{align*}
This implies that $|low(\mb{x}) - low(\y)| > \phi - 2\delta > 0$ as $\phi > 2\delta$ by assumption. Both $low(\x)$ and $low(\y)$ are integer-valued functions, so it must be the case that $low(\x) \neq low(\y)$.

Conversely, suppose that 
\[\lvert \Low{(\mb{x}'; \mathcal{P}_{\texttt{L}})} - \Low{(\mb{y}'; \mathcal{P}_{\texttt{L}})} \rvert \leq \phi.\] Then 
\begin{align*}
    | low(\mb{x}) - low(\mb{y}) | &\leq | \Low{(\mb{x}'; \mathcal{P}_{\texttt{L}})} - low(\mb{x}) | \\
    &+ | \Low{(\mb{y}'; \mathcal{P}_{\texttt{L}})} - low(\mb{y}) | \\
    &+ | \Low{(\mb{x}'; \mathcal{P}_{\texttt{L}})} - \Low{(\mb{y}'; \mathcal{P}_{\texttt{L}})} | \\
    & < \delta + \delta + \phi
\end{align*}
And so we have that $|low(\x) - low(\y)| < 2\delta + \phi < 1$ as $\phi < 1 - 2\delta$. Again, as $low$ is an integer-valued function, it must be the case that $low(\x) = low(\y)$.
\end{proof}

\begin{remark}
Tracing the proof of Lemma \ref{lem:approxlow} also reveals that the intervals on which $\lvert \Low{(\mb{x'}; \mathcal{P}_{\texttt{L}})} - \Low{(\mb{y'}; \mathcal{P}_{\texttt{L}})} \rvert$ and $\phi$ live are disjoint, and so it will never be the case that 
\[|\Low(\x'; \mathcal{P}_{\texttt{L}}) - \Low(\y'; \mathcal{P}_{\texttt{L}})| = \phi,\] despite the statement of the lemma.

 \end{remark}

The implications of \textsc{Lemma}~\ref{lem:approxlow} is that one does not need to be very accurate in the calculation of $\Low(\x'; \mathcal{P}_{\texttt{L}})$, and in fact only needs to approximate $low(\x)$ (using $\Low(\x'; \mathcal{P}_{\texttt{L}})$) to an accuracy of $\frac{1}{4}$. If that condition is guaranteed, then one may compare the value $|\Low(\x'; \mathcal{P}_{\texttt{L}}) - \Low(\y'; \mathcal{P}_{\texttt{L}})|$ to any $2\delta < \phi < 1 - 2\delta$ to check whether the underlying low values are equal or not.

With this lemma, our strategy to compare $low$ values of two approximately binary vectors will be to exploit an approximation of the function that compares the relative size of its two inputs. First, we introduce the following function:

\begin{definition}
For $\mb{x, y} \in \{ 0, 1 \}^n$, let $l_{\x} = low(\x)$ and $l_{\y} = low(\y)$. Define $$lowcomp(l_{\x}, l_{\y}) = \begin{cases} 0, & \text{ if } l_{\x} \neq l_{\y} \\ 1, & \text{ if } l_{\x} = l_{\y} 
\end{cases}. $$
\end{definition}
\noindent The function $lowcomp$ will be used to gate the mod 2 addition of two columns in place of the conditional equality check in Algorithm \ref{alg:standard persistence alg}. In particular, for a given $\mb{x}$ and $\mb{y} \in [0, 1]^n$, the statement ``update $\mb{x}$ to $\mb{x} + \mb{y} \mod 2$, if their lows are equal'' may be reinterpreted as
\begin{align*}
    \mb{x} = \mb{x} + lowcomp(l_{\x}, l_{\y}) \mb{y} \mod 2.
\end{align*}
We now establish a \texttt{LowComp} algorithm to estimate the $lowcomp$ function for approximately binary vectors. Our formulation is based on the \texttt{Comp} algorithm, which estimates the $comp$ function given in Definition \ref{def:comp} (both introduced in \cite{cheon2019}) that compares the relative size of its inputs.


\begin{definition}[\cite{cheon2019}]
\label{def:comp}
For any non-zero real numbers $a,b$, define $$comp(a, b) = \displaystyle \lim_{k \rightarrow \infty} \frac{a^k}{a^k + b^k} = \begin{cases}
1, & \text{ if } a > b \\
\frac{1}{2}, & \text{ if } a = b \\
0, & \text{ if } a < b
\end{cases}$$
\end{definition}

The $\Comp$ algorithm (Algorithm~\ref{alg:Comp}), approximates the $comp$ function by evaluating the expression $\displaystyle \frac{a^{m^t}}{a^{m^t} + b^{m^t}}$, for $t$ a positive integer, and $m$ often chosen to be a power to 2. 


\begin{algorithm}[t]
\caption{\texttt{Comp}($a, b; d, d', m, t$) from \cite{cheon2019}}
\label{alg:Comp}
\flushleft \textbf{Input:} distinct real numbers $a, b \in [1/2, 3/2)$; $d, d', m, t \in \mathbb{N}$\\
\textbf{Output:} a real number $r \in (0, 1)$ \Comment{$r \approx 1$ if $a > b$; $r \approx 0$ if $a < b$}\\
\textbf{Depth:} $d' + 1 + t(d + \log(m) + 2)$ \\
\textbf{Complexity:} $O(d' + t(d +\log(m)))$
\begin{algorithmic}[1]
\State $I \gets \texttt{Inv}(\frac{a + b}{2}; d')$
\State $a_0 \gets \frac{a}{2}\cdot I$
\State $b_0 \gets 1 - a_0$
\For {$n \gets 0 \textbf{ to } t-1$}
    \State $I \gets \texttt{Inv}(a_n^m + b_n^m; d)$
    \State $a_{n+1} \gets a_{n}^m \cdot I$
    \State $b_{n+1} \gets 1 - a_{n+1}$
\EndFor
\State \Return $a_t$
\end{algorithmic}
\end{algorithm}

$\Comp$, along with \textsc{Lemma}~\ref{lem:approxlow}, are the building blocks we need to build $\LowComp$. Using \textsc{Lemma}~\ref{lem:approxlow}, we make the observation that 
\begin{align*}
lowcomp(\x, \y) = 1 & \Leftrightarrow low(\x) = low(\y)\\
& \Leftrightarrow \phi \geq |\Low(\x'; \mathcal{P}_{\texttt{L}}) - \Low(\y'; \mathcal{P}_{\texttt{L}})|\\
& \Leftrightarrow \phi^2 \geq (\Low(\x') - \Low(\y'))^2
\end{align*}
and so we compare $(\Low(\x'; \mathcal{P}_{\texttt{L}}) - \Low(\y'; \mathcal{P}_{\texttt{L}}))^2$ to $\phi^2$ to determine if the underlying $low$ values are equal or not. This construction removes the need to implement an HE circuit to compute absolute value at the cost of two squarings.

We make two important notes before we explicitly define $\LowComp$. The first is that, by construction, $|\Low(\x'; \mathcal{P}_{\texttt{L}}) - \Low(\y'; \mathcal{P}_{\texttt{L}})|$ and $\phi$ exist in disjoint intervals (refer to \textsc{Lemma}~\ref{lem:approxlow}'s remark), and so $\phi$ and $|\Low(\x'; \mathcal{P}_{\texttt{L}}) - \Low(\y'; \mathcal{P}_{\texttt{L}})|$ will never be equal. Thus $\LowComp$ may be treated as an approximate binary indicator function for our application. The second is that the input $(\Low(\x'; \mathcal{P}_{\texttt{L}}) - \Low(\y'; \mathcal{P}_{\texttt{L}}))^2$ is in the interval $[0, (n-1)^2]$. As the $\Comp$ function requires its inputs to be in the interval $[\frac{1}{2}, \frac{3}{2})$, we apply a linear transformation to bring values in the correct interval.

\begin{transformation}
\label{trans:T comp}
    $T_{\texttt{C}}(x) :=  \frac{1}{2} + \frac{x}{n^2}$ 
\end{transformation}

 
 Since $T_{\texttt{C}}$ is a monotonic function, the relative order of the inputs are preserved. We now explicitly define $\LowComp$ by performing $\Comp$ on $T_{\texttt{C}}(\phi^2)$ and $T_{\texttt{C}}((\texttt{L}_{\mb{x}'} - \texttt{L}_{\mb{y}'})^2)$ as described in Algorithm \ref{alg:LowComp}. 

$\LowComp$ inherits from $\Comp$ that its outputs live in (0,1) and that it can approximate $lowcomp$ arbitrarily well given appropriately chosen parameters. We formalize this in the following theorem.

\begin{theorem}
\label{theorem: lowcomp error}
Let $\mb{x}, \mb{y} \in \{0, 1\}^n$ and $\mb{x}', \mb{y}' \in [0, 1]^n$ and assume that $\mathcal{P}_{\texttt{L}}$ is chosen such that $\lvert \Low(\mb{x}'; \mathcal{P}_{\texttt{L}}) - low(\mb{x}) \rvert < \delta$ and $\lvert \Low{(\mb{y}'; \mathcal{P}_{\texttt{L}})} - low(\mb{y}) \rvert < \delta$ for some $0 < \delta < \frac{1}{4}$. Let $\phi$ be any value in the interval $(2\delta, 1-2\delta)$. Define $\LowComp$ as in Algorithm \ref{alg:LowComp}.

If the parameters in the $\Comp$ function are chosen such that 
\[|\Comp(a, b; d, d', m, t) - comp(a, b)| < \eta,\] 
then we also have $$|\LowComp(\textup{\texttt{L}}_{\mb{x}'}, \textup{\texttt{L}}_{\mb{y}'}, \phi; d, d', m, t) - lowcomp(l_{\x}, l_{\y})| < \eta.$$
\end{theorem}
\begin{proof}
See Appendix~\ref{s: LowComp correctness proofs}.
\end{proof}
\begin{algorithm}[t]
\caption{\texttt{LowComp}($\texttt{L}_{\mb{x}'}, \texttt{L}_{\mb{y}'}, \phi; d, d', m, t$)}
\label{alg:LowComp}
\flushleft \textbf{Input:} $\texttt{L}_{\mb{x}'}, \texttt{L}_{\mb{y}'} \in [0, n-1], \phi \in (2\delta, 1-2\delta)$; $d, d', m, t \in \mathbb{N}$ \\
\textbf{Output:} A real number $r \in (0, 1)$ \Comment{$r \approx lowcomp(l_{\x}, l_{\y})$}\\
\textbf{Depth:} $d' + 2 + t(d + \log(m) + 2)$ \\
\textbf{Complexity:} $O(d' + t(d +\log m))$
\begin{algorithmic}[1]
\State $\texttt{L}_{d} \gets \texttt{L}_{\mb{x}'} - \texttt{L}_{\mb{y}'}$
\State $r \gets \Comp(T_{\texttt{C}}(\phi^2), T_{\texttt{C}}(\texttt{L}_{d}^2); d, d', m, t)$ \Comment{Algorithm \ref{alg:Comp}}
\State \Return $r$
\end{algorithmic}
\end{algorithm}
 
\subsection{Parameters for $\LowComp$}

We shall proceed with the analysis of $\LowComp$'s parameters in a similar fashion to $\Low$'s parameters in Section \ref{subsection: low parms}. Theorem 4 from \cite{cheon2019} gives lower bounds for the parameters $d, d', m$, and $t$ to achieve $2^{-\alpha}$ error in the $\Comp$ function.

\begin{theorem}[Theorem 4 in \cite{cheon2019}]
\label{Comp parms}
    Let $x, y \in [1/2, 3/2)$ satisfy $$c \leq max(x, y)/min(x, y)$$ for a fixed $c \in (1, 3)$. If 
    \begin{align*}
    t &\geq \frac{1}{\log (m)}[\log(\alpha + 1) - \log \log (c)] \\
    d &\geq \log(\alpha + t + 2) + m - 2 \\
    d' & \geq \log(\alpha + 2) - 1
    \end{align*}
    then $|\Comp(x, y; d, d', m, t) - comp(x, y)| < 2^{-\alpha}$. 
\end{theorem}


The role of $c$ in $\Comp$ is similar to $\MaxIdx$ in the Section \ref{subsection: low parms}: the closer $c = \max(a, b)/\min(a, b)$ is to 1, the larger the value for all subsequent choice of parameters, thus increasing the ``effort'' needed for the $\Comp$ function to distinguish which of the two inputs is larger. For this reason, it is our goal to bound $c = \max(a, b)/\min(a, b)$ as far from 1 as possible.

Once a $\phi$ is fixed, the only guarantee is that $T_{\texttt{C}}((\texttt{L}_{\x'} - \texttt{L}_{\x'})^2)$ is either strictly greater or less than said $T_{\texttt{C}}(\phi^2)$ (see \textsc{Lemma} \ref{lem:approxlow}'s remark). Since we are only concerned with whether $low(\x)$ and $low(\y)$ are equal or not, the ratio $c$ may be reinterpreted as 
\begin{align*}c = \frac{\max \left\{ T_{\texttt{C}}(\phi^2), T_{\texttt{C}}((\texttt{L}_{\x'} - \texttt{L}_{\x'})^2) \right\}}{\min \left\{ T_{\texttt{C}}(\phi^2), T_{\texttt{C}}((\texttt{L}_{\x'} - \texttt{L}_{\x'})^2) \right\}} > \begin{dcases*} 
\frac{T_{\texttt{C}}(\phi^2)}{T_{\texttt{C}}((\texttt{L}_{\x'} - \texttt{L}_{\x'})^2)}, & \text{ if } $low(\x) = low(\y)$ \\
\frac{T_{\texttt{C}}((\texttt{L}_{\x'} - \texttt{L}_{\x'})^2)}{T_{\texttt{C}}(\phi^2)},  & \text{ if } $low(\x) \neq low(\y)$
\end{dcases*}
\end{align*}

It follows that $$c > \min \left\{ \frac{T_{\texttt{C}}(\phi^2)}{T_{\texttt{C}}((2\delta)^2)} , \frac{T_{\texttt{C}}((1-2\delta)^2)}{T_{\texttt{C}}(\phi^2)} \right\} > 1$$ where the minimum changes depending on which case we are in. Thus, once a $\delta \in (0, 1/4)$ is chosen, this expression is variable with respect to the value of $\phi$ and thus $T_{\texttt{C}}(\phi^2)$. The optimal choice of $\phi$ will ensure the minimum of these two ratios are as far away from $1$ as possible. So, we aim to optimize the right side of this expression with respect to $\phi$: that is, to determine what value of $T_{\texttt{C}}(\phi^2)$ solves 
\begin{equation}
\max \left( \min \left\{ \frac{T_{\texttt{C}}(\phi^2)}{T_{\texttt{C}}((2\delta)^2)}, \frac{T_{\texttt{C}}((1-2\delta)^2)}{T_{\texttt{C}}(\phi^2)} \right\} \right),
\label{eq:optimize_lambda}
\end{equation} 
where the max is taken over $T_{\texttt{C}}(\phi^2)$ in  the interval $\left(T_{\texttt{C}}((2\delta)^2), T_{\texttt{C}}((1-2\delta)^2)\right)$. 

This solution to Eq. \eqref{eq:optimize_lambda} comes from a general fact about positive real numbers, which we prove in \textsc{Proposition}~\ref{prop: geometric mean}, and which establishes the following corollary:

\begin{corollary}
\label{corollary: lowcomp c bound}
The value of $T_{\texttt{C}}(\phi^2)$ which solves Eq. \eqref{eq:optimize_lambda} 
is $$T_{\texttt{C}}(\phi^2)  = \sqrt{\left(\frac{1}{2} + (\frac{2\delta}{n})^2 \right) \left(\frac{1}{2} + (\frac{1-2\delta}{n})^2 \right)}.$$ Thus,  $c > \sqrt{\frac{n^2 + 2(1-2\delta)^2}{n^2 + 2(2\delta)^2}}.$ 
\end{corollary}
\begin{proof}
See Appendix~\ref{s: LowComp parameters proofs}.
\end{proof}

Having determined the bottleneck value $c$, we explicitly construct a choice of parameters for $\LowComp$ to achieve any desired level of accuracy (which has been re-contextualized from the $2^{-\alpha}$ error in $\Comp$ to an arbitrary $\eta$ error in $\LowComp$, see \textsc{Lemma}~\ref{lemma: alpha to eta}). 

\begin{theorem}
\label{theorem: lowcomp parms}
Let $\mb{x}, \mb{y} \in \{0, 1\}^n$ and $\mb{x}', \mb{y}' \in [0, 1]^n$ and assume that $\mathcal{P}_{\texttt{L}}$ is chosen such that $\lvert \Low(\mb{x}'; \mathcal{P}_{\texttt{L}}) - low(\mb{x}) \rvert < \delta$ and $\lvert \Low{(\mb{y}'; \mathcal{P}_{\texttt{L}})} - low(\mb{y}) \rvert < \delta$ for some $0 < \delta < \frac{1}{4}$. Define $\LowComp$ as in Algorithm \ref{alg:LowComp}, where we explicitly pick $$\phi = n\sqrt{\sqrt{\left(\frac{1}{2} + (\frac{2\delta}{n})^2 \right) \left(\frac{1}{2} + (\frac{1-2\delta}{n})^2 \right)} - \frac{1}{2}}.$$

If the parameters in the $\Comp$ function are chosen such that
\begin{align*}
    \alpha &> -\log(\eta) \\
    t &\geq \frac{1}{\log (m)} \left[ \log(\alpha + 2) - \log \log \left( \sqrt{\frac{n^2 + 2(1-2\delta)^2}{n^2 + 2(2\delta)^2}} \right) \right]\\
    d &\geq \log(\alpha + t + 2) + m - 2 \\
    d' &\geq \log(\alpha + 2) - 1  
\end{align*}
then $\LowComp$ has $\eta$-error. That is, $$|\LowComp(\textup{\texttt{L}}_{\mb{x}'}, \textup{\texttt{L}}_{\mb{y}'}, \phi; d, d', m, t) - lowcomp(l_{\x}, l_{\y})| < \eta$$
\end{theorem}
\begin{proof}
See Appendix~\ref{s: LowComp parameters proofs}.
\end{proof}

\begin{remark}
$\LowComp$ can now be thought of as a function of only two inputs ($\texttt{L}_{\x'}$ and $\texttt{L}_{\y'}$), as we will always choose this optimal value of $\phi$.
The theorem also implies a trade-off between $\delta$ and $\eta$. Indeed, estimating $low$ using $\Low$ to a high degree requires less ``effort'' for $\Comp$ to distinguish the (in)equality of two $\Low$ values. Similarly, less accurate $low$ estimates will require $\Comp$ to do more of the heavy-lifting. This intuition is confirmed by the dependence on $\delta$ of the lower bound on $c$. As $\delta$ approaches 0, the bound on $c$ increases further away from 1, causing our choice of parameters for $\LowComp$ to get smaller. On the flip side, as $\delta$ approaches its upper limit of $1/4$, then $c$ may get arbitrarily close to 1, causing $\LowComp$'s parameters to get arbitrarily large. We refer to these parameters as $\Pcomp = (d_{\texttt{C}}, d'_{\texttt{C}}, m_{\texttt{C}}, t_{\texttt{C}})$.
\end{remark}

 
\subsection{Conditional modular addition of vectors}
\label{subsection: mod 2 addition}

For a given $\mb{x}$ and $\mb{y} \in [0, 1]^n$, the statement ``update $\mb{x}$ to $\mb{x} +\mb{y} \mod 2$, if their $low$ values are equal'' from Equation \ref{eq: original modular vector update} may be reinterpreted as
\begin{align}
\label{eq:mod2}
    \mb{x} = \mb{x} + lowcomp(l_{\x}, l_{\y}) \mb{y} \mod 2.
\end{align}
Furthermore, addition modulo 2 can be recast as a polynomial operation using the observation that for any two $a, b \in \{ 0, 1\}$, the operation $(a-b)^2 = a + b \mod{2}$.
Thus, we may rewrite \eqref{eq: original modular vector update} as 
\begin{equation*}
    \mb{x} = lowcomp(l_{\x}, l_{\y})(\mb{x - y})^2 + (1 - lowcomp(l_{\x}, l_{\y}))\mb{x},
\end{equation*}
taking all operations component-wise, to remove mod 2 addition. 

We may then approximate this operation using $\Low$ to esimate $low$ and $\LowComp$ to estimate $lowcomp$. That is, the operation we will be performing on approximate binary vectors is
\begin{equation}
    \label{eq: gated mod 2 addition}
    \mb{x}' = \LowComp(\textup{\texttt{L}}_{\mb{x}'}, \textup{\texttt{L}}_{\mb{y}'})(\mb{x' - y'})^2 + (1 - \LowComp(\textup{\texttt{L}}_{\mb{x}'}, \textup{\texttt{L}}_{\mb{y}'}))\mb{x'},
\end{equation}
as alluded to in Eq. \eqref{eq:approximating_column_update} in Section~\ref{subsection: LowComp}.
\subsection{Modifying the Logical Structure of \texttt{Reduce}}
\label{s: ModifiedLogic}

The main operation in $\Reduce$ (Algorithm~\ref{alg:standard persistence alg}) is gated by a conditional \textbf{while} loop. As mentioned before, conditional statements cannot be explicitly implemented (or traced) over ciphertexts. Therefore, we need to rewrite lines 2-6 
in Algorithm~\ref{alg:standard persistence alg} so that they are HE-compatible.
This is done by replacing the \textbf{while} loop with a double nested \textbf{for} loop, each of which run through all preceding column indices. Assume $low(\mb{\Delta}_i) \neq low(\mb{\Delta}_k)$ for all $0 \leq i \neq k < j$, as is the case when the $\Reduce$ algorithm, applied to a boundary matrix $\mb{\Delta}$, first encounters column $j$. If we loop through the preceding $j$ columns once, comparing each $\mb{\Delta}_k, k=0\ldots, j-1$ to $\mb{\Delta}_j$, either $low(\mb{\Delta}_k) = low(\mb{\Delta}_j)$ for some $k<j$ or not. In the latter case, we know $\mb{\Delta}_j$ is already in reduced form and will not change---no matter how many times one loops again through the preceding $j-1$ columns---since the (binary) addition only happens when the $low$'s of two columns match. On the other hand, if some  $low(\mb{\Delta}_k) = low(\mb{\Delta}_j)$ for some $k<j$, $\mb{\Delta}_j \leftarrow \mb{\Delta}_j + \mb{\Delta}_k \mod 2$ will change $\mb{\Delta}_j$, and in particular, this addition necessarily causes the $low(\mb{\Delta}_j)$ to decrease. Thus, after such an update this new $\mb{\Delta}_j$ will never again be equal to $\mb{\Delta}_k$. In other words each column vector will only update column $j$ at most once. Without any assumptions about the order in which preceding columns update $\mb{\Delta}_j$, we simply loop over the preceding columns enough times to guarantee every vector which should have updated column $j$ has done so. This requires exactly $j$ loops over all preceding columns since each preceding column can only update $\mb{\Delta}_j$ at most once. For the base case, note that column $j=0$ is trivially in reduced form and $\mb{\Delta}_1$ will certainly be in reduced form after a single comparison with $\mb{\Delta}_0$. This aligns with the worst case complexity for the original $\Reduce$ algorithm: $O(j^2)$ for column $j$, $O(n^3)$ overall \cite{edelsharer}.
In Section \ref{s: Secure low}, we modified the existing $\MaxIdx$ from \cite{cheon2019} to attain the $\Low$ algorithm to estimate $low$. We have already discussed how to check the equality of two $low$ values using $\Low$ and $\LowComp$ in Section \ref{subsection: LowComp}. Finally, the mod 2 addition over rational numbers was constructed in Section \ref{subsection: mod 2 addition}. With all of this combined, we may now rewrite the main block of the \texttt{Reduce} algorithm as written in lines 6-9 of Algorithm \ref{alg:secure persistence alg} in a way which makes it compatible with each approximate algorithm on approximate vectors framework.

  
\subsection{An HE-Compatible \texttt{Reduce}}
\label{s: HE-Reduce-Alg}
As the challenges as listed
in Section~\ref{subsection: HE}
have now been addressed in Sections~\ref{s: Secure low}-\ref{s: ModifiedLogic}, we now present Algorithm \ref{alg:secure persistence alg}, which 
 is an HE-compatible version
of $\scReduce$ and which can take an encrypted boundary matrix as input and reduce it using HE-operations in the ciphertext space.
\begin{algorithm}[t]
\caption{\texttt{HE-Reduce}($\mb{\Delta}; \mathcal{P}_{\texttt{L}}, \mathcal{P}_{\texttt{C}}$)}
\label{alg:secure persistence alg}
\flushleft \textbf{Input:} A boundary matrix $\mb{\Delta} = [\: \mb{\Delta}_0 \: | \: \mb{\Delta}_1 \: | \: ... \: | \: \mb{\Delta}_{n-1} \: ] \in \mathbb{Z}_2^{n \times n}$ \\
\textbf{Output:} A matrix $\mb{R}' \in [0, 1]^{n \times n}$ \Comment{Approximates Algorithm \ref{alg:standard persistence alg}}
\begin{algorithmic}[1]

\State $\textup{\texttt{L}}_0 \gets \Low(\mb{\Delta}_0; \mathcal{P}_{\texttt{L}})$ \Comment{Algorithm \ref{alg:low}}
\State $\mb{\Delta}'_0 \gets \mb{\Delta}_0$

\For {$j \gets 1$ \textbf{to} $n-1$}
    \State $\textup{\texttt{L}}_j \gets \Low(\mb{\Delta}_j; \mathcal{P}_{\texttt{L}})$ \Comment{Algorithm \ref{alg:low}}
    \State $\mb{\Delta}'_j \gets \mb{\Delta}_j$
    \For {$k \gets 0$ \textbf{to} $j-1$} \Comment{Section \ref{s: ModifiedLogic}}
        \For {$j_0 \gets 0$ \textbf{to} $j-1$} 
        \State $\Omega \gets \LowComp(\textup{\texttt{L}}_{j_0}, \textup{\texttt{L}}_j; \mathcal{P}_{\texttt{C}})$ \Comment{Algorithm \ref{alg:LowComp}}
        \State $\mb{\Delta}'_{j} \gets \Omega  (\mb{\Delta}'_j - \mb{\Delta}'_{j_0})^2 + (1 - \Omega)  \mb{\Delta}'_j$ \Comment{Equation \ref{eq: gated mod 2 addition}}
        \State $\textup{\texttt{L}}_j \gets \Low(\mb{\Delta}'_j; \mathcal{P}_{\texttt{L}})$ \Comment{Algorithm \ref{alg:low}}
        \EndFor
    \EndFor
\EndFor
\State \Return $\mb{R}' = [\: \mb{\Delta}'_0 \: | \: \mb{\Delta}'_1 \: | \: ... \: | \: \mb{\Delta}'_{n-1} \: ]$
\end{algorithmic}
\end{algorithm}
We note that the moment we do the very first column addition, vectors have moved from $\{0, 1\}^n$ to $[0, 1]^n$, requiring the need for all algorithms to be compatible with approximate binary vectors. For this reason, we must have a guarantee of correctness, which is a function of the controllable errors in our approximation variables: $\mb{v}'$ (\textsc{Theorem}~\ref{theorem: lowerror}), $\Low$ (\textsc{Lemma}~\ref{lem:approxlow}) , and $\LowComp$ (Section \ref{subsection: low parms}). 


As long as $|\mb{v}' - \mb{v}| < 1/2n$, we know that $\Low(\mb{v}'; \mathcal{P}_{\texttt{L}})$ will approximate $low(\mb{v})$ as accurately as wanted. And as long as $\Low(\mb{v}'; \mathcal{P}_{\texttt{L}})$ estimates $low(\mb{v})$ to within $1/4$, $\LowComp$ is able to distinguish between $low(\x)$ and $low(\y)$ using $\Low(\mb{x}'; \mathcal{P}_{\texttt{L}})$ and $\Low(\mb{y}'; \mathcal{P}_{\texttt{L}})$. $\LowComp$ directly defines an approximately binary indicator 
$$\Omega \gets \LowComp(\textup{\texttt{L}}_{j_0}, \textup{\texttt{L}}_j; \mathcal{P}_{\texttt{C}})$$ which will be used to perform the ``mod 2'' addition, which will naturally have accumulating non-zero errors (determined by $\eta$). The finiteness of the algorithm guarantees the existence of an $\eta$ such that the accumulation of errors never exceeds the maximum threshold of $1/2n$. In a strict sense, $\HEReduce$ only fails to produce the correct reduced boundary matrix if the maximum error in some component is 1/2 or larger. If $|\Reduce(\mb{\Delta}) - \mb{R}'| < 1/2$, then $\texttt{Round}(\mb{R}') = \Reduce(\mb{\Delta}),$ where $\texttt{Round}$ casts entries to the nearest integer. This condition is guaranteed by the stricter requirement that errors are within $1/2n$.

\section{Complexity and Implementation Analysis}
\label{s: Complexity}
As in all HE-compatible functions, there is particular interest in $\HEReduce$'s complexity and depth to understand the noise growth that a ciphertext will accumulate as it passes through the algorithm. We will prove a more general statement that establishes the depth of our algorithm on an $n \times n$ boundary matrix. We note that while we establish the textbook version of the $\Reduce$ algorithm as $\HEReduce$, an immediate improvement to the algorithm to make it even more HE-compatible is easily seen. We implement this verison, $\HEReduceOptimized$, and analyze its depth and complexity.

\subsection{Analysis of $\HEReduceOptimized$}

\label{ss: HE-Optimized-Reduce}

In our implementation, we use Algorithm~\ref{alg:secure persistence alg optimized}, which is a slightly modified version of Algorithm~\ref{alg:secure persistence alg}. Here, the \texttt{Low} computation in line~10 in 
Algorithm~\ref{alg:secure persistence alg} is now pushed out of the \textbf{for} loop (see line~14 in Algorithm~\ref{alg:secure persistence alg optimized}) and the repetitive update operations
\[\mb{\Delta}'_{j} \gets \Omega   (\mb{\Delta}'_j - \mb{\Delta}'_{j_0})^2 + (1 - \Omega)   \mb{\Delta}'_j,\ j_0=0,...,j-1\]
in line~9 in Algorithm~\ref{alg:secure persistence alg} are now replaced by a single cumulative update operation (line~13 in Algorithm~\ref{alg:secure persistence alg optimized}), which can be explicitly rewritten as
\begin{equation}
\label{eq: cumulative low}
\mb{\Delta}'_{j} \gets \sum_{j_0=0}^{j-1}{\Omega_{j_0,j}((\mb{\Delta}'_j - \mb{\Delta}'_{j_0})^2)} + (1 - \sum_{j_0=0}^{j-1}{\Omega_{j_0,j}})   \mb{\Delta}'_j,
\end{equation}
where $\Omega_{j_0,j} = \LowComp(\textup{\texttt{L}}_{j_0}, \textup{\texttt{L}}_j; \mathcal{P}_{\texttt{C}})$.
The correctness follows from the fact that 
$\Omega_{j_0,j}$ is either approximately zero for all $j_0=0,...,j-1$ except for at most one value of $j_0=k$ (where it is approximately one), whence we have either
$\mb{\Delta}'_{j}$ stays approximately the same 
or is updated to
$\mb{\Delta}'_{j} \approx (\mb{\Delta}'_j - \mb{\Delta}'_{k})^2$,
as required.

\begin{algorithm}[t]
\caption{\texttt{HE-Reduce-Optimized}($\mb{\Delta}; \mathcal{P}_{\texttt{L}}, \mathcal{P}_{\texttt{C}}$)}
\label{alg:secure persistence alg optimized}
\flushleft \textbf{Input:} A boundary matrix $\mb{\Delta} = [\: \mb{\Delta}_0 \: | \: \mb{\Delta}_1 \: | \: ... \: | \: \mb{\Delta}_{n-1} \: ] \in \mathbb{Z}_2^{n \times n}$ \\
\textbf{Output:} A matrix $\mb{R}' \in [0, 1]^{n \times n}$ \Comment{Approximates Algorithm \ref{alg:standard persistence alg}}
\begin{algorithmic}[1]
\State $\textup{\texttt{L}}_0 \gets \Low(\mb{\Delta}_0; \mathcal{P}_{\texttt{L}})$ \Comment{Algorithm \ref{alg:low}}
\State $\mb{\Delta}'_0 \gets \mb{\Delta}_0$
\For {$j \gets 1$ \textbf{to} $n-1$}
    \State $\textup{\texttt{L}}_j \gets \Low(\mb{\Delta}_j; \mathcal{P}_{\texttt{L}})$ \Comment{Algorithm \ref{alg:low}}
    \State $\mb{\Delta}'_j \gets \mb{\Delta}_j$
    \For {$k \gets 0$ \textbf{to} $j-1$} \Comment{Section \ref{s: ModifiedLogic}}
        \State $\texttt{cumLeft} = 0,\ \texttt{cumOmega} = 0$
        \For {$j_0 \gets 0$ \textbf{to} $j-1$} 
        \State $\Omega \gets \LowComp(\textup{\texttt{L}}_{j_0}, \textup{\texttt{L}}_j; \mathcal{P}_{\texttt{C}})$ \Comment{Algorithm \ref{alg:LowComp}}
        \State $\texttt{cumOmega} \gets \texttt{cumOmega} + \Omega$
        \State $\texttt{cumLeft}\gets \texttt{cumLeft} + \Omega   (\mb{\Delta}'_j - \mb{\Delta}'_{j_0})^2$
        \EndFor
        \State $\mb{\Delta}'_{j} \gets \texttt{cumLeft} + (1 - \texttt{cumOmega})   \mb{\Delta}'_j$ \Comment{Equation \ref{eq: cumulative low}}
        \State $\textup{\texttt{L}}_j \gets \Low(\mb{\Delta}'_j; \mathcal{P}_{\texttt{L}})$
    \EndFor
\EndFor
\State \Return $\mb{R}' = [\: \mb{\Delta}'_0 \: | \: \mb{\Delta}'_1 \: | \: ... \: | \: \mb{\Delta}'_{n-1} \: ]$
\end{algorithmic}
\end{algorithm}

\begin{theorem}
\label{thm: depth and complexity optimized}
    Let $\mathbf{B} \in \mathbb{Z}_2^{n \times m}$ be a binary matrix with $n \geq m$. Furthermore suppose the tuples of parameters $\Plow = (d_{\texttt{L}}, d'_{\texttt{L}}, m_{\texttt{L}}, t_{\texttt{L}})$ and $\Pcomp = (d_{\texttt{C}}, d_{\texttt{C}}', m_{\texttt{C}}, t_{\texttt{C}})$ are given which give depth $D_{\texttt{L}} = d_{\texttt{L}} + 1 + t_{\texttt{L}}(d'_{\texttt{L}} + \log(m_{\texttt{L}}) + 2)$ and $D_{\texttt{C}} = d_{\texttt{C}} + 1 + t_{\texttt{C}}(d'_{\texttt{C}} + \log( m_{\texttt{C}}) + 2)$ to the $\Low$ and $\Comp$ functions, respectively. 
    Then, the depth of the $\HEReduceOptimized$ (Algorithm \ref{alg:secure persistence alg optimized}) is $\frac{m(m-1)}{2}[D_L + D_C + 1]$ and its complexity is
    \begin{align*}
\mathcal{O} \big(m^3[1 + d'_{\texttt{C}} + t_{\texttt{C}}(d_{\texttt{C}} + \log(m_{\texttt{C}}))] + m^2[d'_{\texttt{L}} + t_{\texttt{L}} (d_{\texttt{L}} + m \log (m_{\texttt{L}}))] \big).
\end{align*}
\end{theorem}

\begin{proof}
We proceed with an induction on $m$. For the base case, note that column $j=0$ is trivially in reduced form and $\mb{\Delta}_1$ will certainly be in reduced form after a single comparison with $\mb{\Delta}_0$. This aligns with the worst case complexity for the original $\Reduce$ algorithm: $O(j^2)$ for column $j$, $O(n^3)$ overall \cite{edelsharer}.

For the inductive hypothesis, assume that for all $j \leq m-1$, that the depth of the $\HEReduceOptimized$ algorithm, after termination on a $n \times j$ matrix, is $d(j) = \frac{j(j-1)}{2} D$. 

Now consider an $n \times m$ matrix $\mathbf{B} = [\: \x_0 \: | \: ... \: | \: \x_{m-2} \: | \: \x_{m-1} \: ] \in \mathbb{Z}_2^{n \times m}$. Then the sub-matrix $\mathbf{B}'$ obtained by excluding the last column $\x_{m-1}$ is an $n \times (m-1)$ matrix, and thus has depth $d(m-1) = \frac{(m-1)(m-2)}{2} D$ by the inductive hypothesis. Let us now focus on the last column, $\mb{x}_{m}$. Consider the outer loop corresponding to $k = 0$ in the $\HEReduceOptimized$ algorithm. After the first inner loop finishes, we have the depth of column $\mb{x}'_{m}$ is exactly $d(m-1) + D = \frac{(m-1)(m-2)}{2} D + D$, where the last $D$ term is added from the very last update.

However, in $\HEReduceOptimized$, for all subsequent $k = 1, ..., m-1$, every $k$ loop adds exactly $D$ to the depth only one time. This is because every run of the inner $j_0$ \textbf{for} loop runs in parallel with ciphertexts of lower depth than the most recent update of $\mb{x}'_m$. A counting argument will yield that the depth of column $\mb{x'}_m$ after all loops are completed is $[\frac{(m-1)(m-2)}{2} D + D] + [(m-2)D] = \frac{m(m-1)}{2}D,$ thus completing the induction.


As for the complexity, the optimized algorithm calls $\Low$ (which has complexity $O(m + d_L' + t_L(d_L + m \log m_L))$ exactly $m(m-1)/2$ times but still calls $\LowComp$ (which has complexity $O(d_C' + t_C(d_C +\log m_C))$) exactly $m(m-1)(2m-1)/6$ times. Thus, the overall complexity is as stated. 
\end{proof}

\begin{remark}
This algorithm performed on a boundary matrix $\mb{\Delta} \in \mathbb{Z}_2^{n \times n}$ is depth $n(n-1)/2 \left[ D_{\texttt{L}} + D_{\texttt{C}} + 1 \right]$ and cost $\mathcal{O}(n^3 + n^2[d'_{\texttt{L}} + t_{\texttt{C}}(d_{\texttt{L}} + n)])$ for the choice of $m_\texttt{C} = m_\texttt{L} = 2$, and assuming that $ d'_{\texttt{L}} >  d'_{\texttt{C}}$, $d_{\texttt{L}} >  d_{\texttt{C}}$, and $t_{\texttt{C}} \approx t_{\texttt{L}}$.
\end{remark}

\subsection{Implementation Notes}
\label{s: implementation notes}
In this section, we discuss our implementation of Algorithm~\ref{alg:secure persistence alg} using HE. We assume that a \texttt{Client} generates \texttt{pk}, \texttt{sk}, \texttt{evk}, for some suitable \texttt{params}; and the \texttt{Server} knows \texttt{pk} and \texttt{evk}. Note that the \texttt{Server} can evaluate circuits on ciphertexts but cannot decrypt; see Section~\ref{subsection: HE}.
By construction, variables of 
Algorithm~\ref{alg:secure persistence alg} deals with vectors over the set $\mathbb{R}$ of real numbers  and the approximate arithmetic is performed over $\mathbb{R}$. Additionally, as comparisons feature heavily in our implementation, we note that CKKS comparison circuits are comparable in amortized time to both BFV/BGV and TFHE schemes \cite{HE-Comparisons}. Therefore, the HEAAN \cite{CKKS} HE scheme, also known as the CKKS scheme, would be a suitable choice for implementing Algorithm~\ref{alg:secure persistence alg}. 
In CKKS, we have $\mathcal{M} = \mathbb{Z}[X]/\langle X^N+1\rangle$ and $\mathcal{C} = \mathbb{Z}_Q[X]/\langle X^N+1\rangle \times \mathbb{Z}_Q[X]/\langle X^N+1\rangle$. 
Moreover, CKKS allows one to encode and encrypt $N/2$ numbers $[x_0,...,x_{N/2 -1}]$,  $x_i\in \mathbb{R}$, as a single ciphertext, where ciphertext operations can be performed component-wise and simultaneously. As a result, under the setting of above CKKS parameters, a \texttt{Client} can encode and encrypt an $n\times n$ boundary matrix $\mb{\Delta}$ in at least two different ways: as $n$ ciphertexts $c_{0},...,c_{n-1}$, where
$c_i\in \mathcal{C}$ represents the encryption of the $i$'th column of $\mb{\Delta}$, which requires $n\le N/2$; or as a single ciphertext $c$, where $c\in \mathcal{C}$ represents the encryption of the ``concatenated columns of $\mb{\Delta}$''-vector, which requires $n\le \sqrt{N/2}$. 

For simplicity, we assume that a \texttt{Client} encrypts $\mb{\Delta}$ using the first method; obtains and sends $c_i\in\mathcal{C}$ to the \texttt{Server}. The \texttt{Server} can use $\texttt{evk}$
and compute $c_0^{\prime}, ..., c_{n-1}^\prime \leftarrow \texttt{Eval}_{\texttt{evk}}(f; c_0, ..., c_{n-1})$, using ciphertext addition and multiplication operations\footnote{In practice, one would have to utilize other ciphertext operations like \texttt{Rotate}.}, where $f$ is the arithmetic circuit induced by Algorithm~\ref{alg:secure persistence alg}. The \texttt{Server} sends $c_i^\prime$, $i=0,...,n-1$, back to the \texttt{Client}, who can use \texttt{sk} and decrypt $c_i^\prime$ to $x_i^\prime$. Note that, by our previous arguments following Algorithm~\ref{alg:secure persistence alg}, \texttt{Round}($x_i^\prime$) would match the $i$th column of \texttt{Reduce}($\mb{\Delta}$).

In order to get a more concrete sense of the implementation of Algorithm~\ref{alg:secure persistence alg} using CKKS, we consider CKKS
parameters at $\lambda = 128$-bit security level, and set $N = 2^{17}$ and $Q = P\cdot q_0\cdot \prod_{i=1}^{50}{q_i}$, as a product of $52$ primes with $\log_2{Q} = 3300$, $\log_2{P} = 660$, and $\log_2{q_i} \approx \delta = 51 < \log_2{q_0} < 2\delta = 102$; see Table 6 in \cite{CKKSParams}. This choice maximizes the depth $L$ of circuits that HEAAN can evaluate, without bootstrapping, to $L = 50$ and the precision digits of data during computations is kept at $10$. 
Under this choice of parameters, a \texttt{Client} can encode and encrypt boundary matrices of size $(n\times n)$, where $n \le N/2 = 2^{16} (\sqrt{N/2} = 2^{8})$ using the first (second) encoding approach. CKKS can handle circuits of depth up to $L=50$ and so one would have to bootstrap \cite{CKKS} once the depth limit is exhausted. 
In our implementation, we use Algorithm~\ref{alg:secure persistence alg optimized}, which is a slightly modified and optimized version of Algorithm~\ref{alg:secure persistence alg}.
Our implementation, using Intel(R) 16-Core(TM) i9-9900K 3.60GHz, can reduce a single encrypted 3x3 matrices in $4.5$ seconds with $40$ bootstrappings
using the (non-cryptographic) CKKS parameters $N=2^5$, $Q\approx 2^{3188}$; and 
$\Plow = \Pcomp = (5, 5, 2, 5)$, where the parameters are chosen such that the underlying $\Comp$ in our computations uses one of the optimal parameters as reported in \cite{cheon2019}. Note that  Reducing 3x3 matrices takes $225$ minutes  using $128$-bit secure CKKS parameters with $N=2^{17}$.
Note that if ciphertext slots are fully utilized then the amortized times would be $4.5/(2^4/(3+1)) = 1.125$ and $225\cdot60/(2^{16}/(3+1)) = 0.82$ seconds, respectively.

\subsection{Limitations and Potential Improvements}
\label{s: Limitations}
A major challenge in implementing \texttt{HE-Reduce} using HE is the cubic co-factor $n^3$ in the depth of the underlying arithmetic circuit (even $\HEReduceOptimized$ has a quadratic co-factor $n^2$, see Theorem~\ref{thm: depth and complexity optimized}.) 
As pointed out in an implementation scenario in Section~\ref{s: implementation notes}, HEAAN can handle circuits up to depth $50$ but the depth of $\HEReduceOptimized$ quickly reaches large numbers as $n$ grows and exceed $50$ even for small values of $n$. Therefore, the size of the boundary matrix may be too large in practice to be practically reducable. Indeed, the Vietoris-Rips and \v{C}ech filtrations have $2^{m}$ simplices in the worst case for a point cloud with $m$ points \cite{PH_roadmap} since they define scales for every simplex in the powerset of the vertices (although it would be unusual to compute with simplices of all dimensions). 
Another challenge is to encode and encrypt ($n\times n$) boundary matrices for large $n$. As noted in Section~\ref{s: implementation notes}, currently suggested HEAAN parameters \cite{CKKSParams} at $128$-bit security level limits $n < N/2 = 2^{16}$ or $n < \sqrt{N/2} = 2^8$, depending on the choice of encoding. Therefore, substantial improvements would be required before an efficient implementation of $\HEReduceOptimized$ can be realized.

A possible improvement would be to reduce the size of the boundary matrix by the choice of filtration, which is an active field of research. For example, for a point cloud of size $m$ in dimension $d$, the (weighted) alpha \cite{Edelsbrunner1995, edelsharer} and sparse Rips filtrations \cite{sparserips} create complexes of size $m^{\mathcal{O}(d/2)}$ and $\mathcal{O}(m)$ respectively \cite{PH_roadmap}. Very recent theoretical results also justify computing the PH of small subsets of a point cloud to construct a distribution of PDs representing the topology of the original cloud \cite{solomon2022}. This approach has the potential to massively reduce the size of each boundary matrix, whose reductions can be carried out completely in parallel.
Another improvement would come from relaxing our theoretical bounds for parameters to reduce the depth in Theorem~\ref{thm: depth and complexity optimized}. Section~\ref{sec:empirical} provides some motivating evidence of the feasibility and potential consequences of such an approach. 

\subsection{Empirical Results}
\label{sec:empirical}
The output of Algorithm~\ref{alg:secure persistence alg optimized} is an approximately binary matrix $$\mb{R}'=\HEReduceOptimized(\mb{\Delta};\Plow,\Pcomp) \in [0,1]^{n\times n}$$ which approximates the output of $\Reduce(\mb{\Delta})$. The key bound in parameter selection is that throughout $\HEReduceOptimized$, the approximate binary vectors must never disagree with the true underlying binary vectors by more than $1/2n$, to ensure the output of $\HEReduceOptimized$ returns an approximately binary vector with the same implied birth-death pairings as the exact $\Reduce$.
How prevalent are the cases in which the maximum error between the approximate and the exact reduced matrix exceeds $1/2n$? This question focuses on the accumulation of error throughout $\HEReduceOptimized$ due to approximating exact operations in plaintext, and is independent of the noise growth that is accumulated by HE operations. 

We explored this question in a fashion similar to the parameter relaxation experiment conducted in \cite{cheon2019} by systematically increasing the parameters $\Plow$ and $\Pcomp$ of $\HEReduceOptimized$ with respect to their depth and complexity to determine a minimum depth cofactor $D = D_L + D_C + 1$ (as defined in Theorem \ref{thm: depth and complexity optimized}) which resulted in 100\% accuracy. Specifically, for each parameter choice, we randomly sampled the space of $10\times 10$, upper-triangular, binary matrices and compared the results of exact and approximate reductions, recording when all entries were within $1/2n$ and/or 1/2 of the exact-reduced binary matrix. We found that the minimum depth (119) and complexity (55300) parameter pair for which 100\% of the approximately reduced matrices were within $1/2n$ of their exact counterparts was $\Plow = (3, 3, 2, 6)$ and $\Pcomp = (3, 3, 2, 12)$, as reported in Table \ref{tab: practical parms}. That said, it may be that some matrices will exhibit an error in excess of the $1/2n$ tolerance for these parameter choices, although we expect such examples to be rare if they exist. By reducing $t_{\texttt{C}}$ from 12 to 11, we found only 81.2\% of approximately reduced matrices had errors less than $1/2n$ and only 91.2\% of matrices had maximum error was less than $1/2$---and so would still yield the correct reduced matrix after rounding (Table \ref{tab: practical parms}). By additionally raising $t_L$ from 6 to 7 (so the circuit depth is again 119 but complexity is 51800) we find 98.6\% of approximately reduced matrices had errors less than $1/2n$ and 100\% of matrices had maximum error was less than $1/2$ (Table \ref{tab: practical parms}). These results in expected accuracy suggests there is moderate sensitivity to the choice of some parameters. 

\begin{table}[t]
\centering
\caption{Empirically-determined parameters for accurate reduction of $10 \times 10$ matrices.  $\mb{*}$ represents the lowest depth and complexity parameters  of $\HEReduceOptimized$  that exhibited correct reduction (within $1/2n$ error) of 100\% of randomly chosen matrices.}
\begin{tabularx}{\textwidth}{cYYYYYY}
 \hline
 & $\Plow$ & $\Pcomp$ & $D$ & Complexity & Within $1/2n$ & Within $1/2$\\ 
 \hline
 & (3, 3, 2, 6) & (3, 3, 2, 11) & 113 & 51300 & 81.2\% & 91.2\% \\ 
 & (3, 3, 2, 7) & (3, 3, 2, 11) & 119 & 51800 & 98.6\% & 100\% \\ 
$\mb{*}$ & (3, 3, 2, 6) & (3, 3, 2, 12) & 119 & 55300 & 100\% & 100\% \\
\hline
\end{tabularx}
\label{tab: practical parms}
\end{table}


We found that the same parameters for \HEReduceOptimized \; shown to correctly reduce random $10\times 10$ matrices, also correctly reduces the $12\times 12$ example boundary matrix given in Figure \ref{fig: square filtration}. Indeed, the maximum error in any component of the approximate reduced boundary matrix is $\num{2.04e-3}$, well within the required $1/2n = 1/24 \approx 0.041$ tolerance to guarantee correct computation of column low 1s (Figure \ref{fig: approximate reduction errors} (A)). 

By relaxing some choices of accuracy parameters we observe failure cases where \HEReduceOptimized \; produces approximate binary matrices that do not cast to the exact reduced matrix. For instance, relaxing $t_{\texttt{L}}$ from 5 to 6 returns a matrix that fails to be in reduced form, as both columns $ab$ and $bc$ have the same low 1s (Figure \ref{fig: approximate reduction errors} (B)). By increasing $t_{\texttt{C}}$ substantially, this issue is remedied, however, the approximately reduced matrix does not agree with the exact reduction (Figure \ref{fig: approximate reduction errors} (C)). It is interesting to note that, in this case, the low 1s are all correct, and so the correct persistence diagram is computed. 

Relaxing $\LowComp$ parameters also leads to failure, as shown in (Figure \ref{fig: approximate reduction errors} (D)), where large errors accumulate during reduction leading to values that fall far outside the allowed range of $[0, 1]$.  

\begin{figure}[t]
    \centering
    \includegraphics[width = 1.0\linewidth]{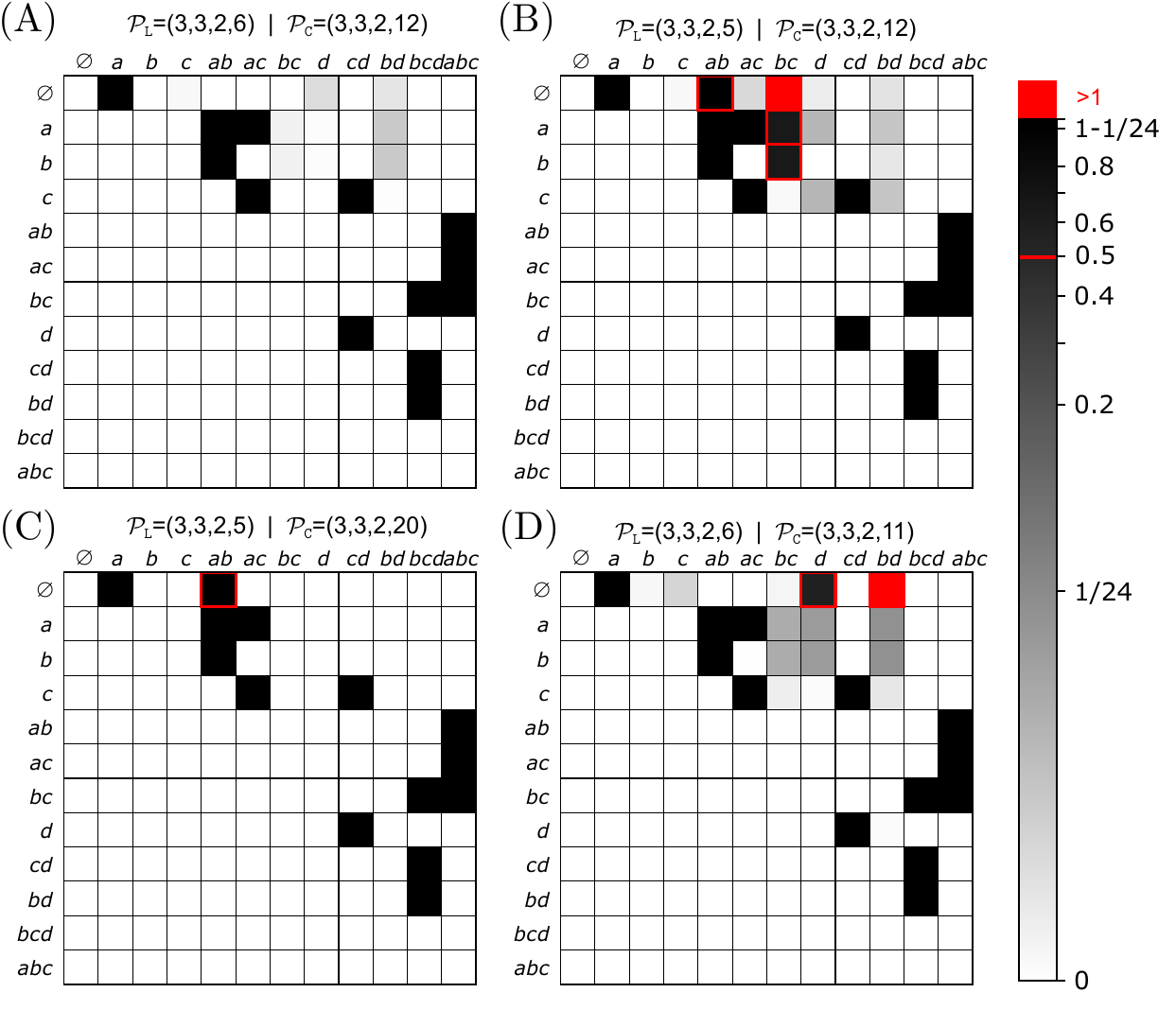}
    \caption{Approximate reductions of the boundary matrix in Figure \ref{fig: square filtration} using arithmetic circuits for various choices of algorithm parameters. Matrix entries are colored according their magnitude after approximate reduction by \HEReduceOptimized. Errors in excess of 1/2 are bordered in red. Entries with values in excess of 1 are colored red.}
    \label{fig: approximate reduction errors}
\end{figure}



\section{Concluding Remarks and Future Research}
\label{s: Concluding remarks}
We developed a new algorithm that enables key TDA computations in the ciphertext space using HE. We proved the correctness of our proposed algorithm, provided detailed correctness and complexity analysis, and an implementation of our algorithms using CKKS from the OpenFHE library \cite{OpenFHE}. We also presented some concrete directions for improvement and provided experimental results. To our knowledge, this is the first attempt to introduce secure computing for TDA. It would be interesting to extend and improve our results, and to implement secure TDA algorithms on realistic data sets.

The $\Reduce$ algorithm represents one of several fundamental components of TDA machinery which challenge existing technologies in the HE space. Another is the calculation of distances between PDs, which rely on combinatorial optimization algorithms to minimize the cost of matchings between persistence pairs in pairs of PDs \cite{BottleneckComplexity}. Others include the numerous methods being broadly deployed to vectorize PDs for use with downstream ML models \cite{pis, landscapes, Perea2022, Chung2022, pmlr-v108-carriere20a, Reininghaus2015ASM}. HE-compatible implementations could allow remote processing of encrypted PDs and would immediately enable the use of existing implementations of encrypted Euclidean distance calculations \cite{SEAL} and encrypted ML models that take as input finite-dimensional feature vectors \cite{HE-LogReg, Secure-Federated, EricLogReg2020}. We are hopeful these challenges will have implications beyond TDA-ML use cases by soliciting contributions from the broader HE community, and that the constraints imposed by HE will motivate new TDA approaches.

\bibliographystyle{splncs04}
\bibliography{ASIACRYPT}

\newpage
\appendix
\noindent {\Large \textbf{Appendix}}
\section{$\Low$ Correctness Proofs}
\label{s: Low correctness proofs}

\begin{appendixproof}[\textbf{Proof of \textsc{Lemma}~\ref{lemma:low_maxidx_exact}}]
First note that $S(\mb{v})\in \mathcal{D}^n$, since all entries are necessarily distinct by definition of Transformation \ref{trans:S}, and so $maxidx(S(\mb{v}))$ is defined. 

Suppose that $low(\mb{v})=k$ for some $0\le k \le n-1$. We have to show
$S(\mb{v})[k]>S(\mb{v})[j]$ for all $0\le j\le n-1$ and $j\ne k$.

Case 1: $0\le j <k$. Note that $\mb{v}[k]=1$ and that $\mb{v}[j]\in \{0,1\}$. Therefore,  
\[S(\mb{v})[k] = \mb{v}[k]+k/n = 1 + k/n > \mb{v}[j] + j/n = S(\mb{v})[j].\]

Case 2: $0\le k<j\le n-1$. Note that $\mb{v}[j]=0$ because $low(\mb{v})=k$ is the largest index with $\mb{v}[k]=1$. Therefore, $S(\mb{v})[j] = \mb{v}[j] + j/n = j/n$ and that
\[S(\mb{v})[k] = 1+k/n > (n-1)/n \ge j/n = S(\mb{v})[j].\]

\end{appendixproof}

\begin{remark}
The same argument in the proof of \textsc{Lemma}~\ref{lemma:low_maxidx_exact} would work for any transformation $S$ that is strictly monotonically increasing on $\{0,1,\ldots, n-1\}$ and strictly bounded by $1$. However, in the implementation of \texttt{MaxIdx}, which we use to approximate $maxidx$ (and thus $low$), the rate of convergence is increasing with the distance between distinct values in $\mb{v}$, and so a nonlinear choice of $S$ may have the effect of decreasing the rate of convergence, at least for some input vectors. Further analysis would be required to find an optimal choice for $S$.
\end{remark}

\begin{appendixproof}[\textbf{Proof of \textsc{Lemma}~\ref{lemma:low_maxidx_approx}}]
Suppose that $low(\mb{v})=k$ for some $0\le k \le n-1$. We have to show
$S(\mb{v'})[k]>S(\mb{v'})[j]$ for all $0\le j\le n-1$ and $j\ne k$.

Case 1: $0\le j <k$. Note that $\mb{v}[k]=1$, $\mb{v'}[k]>1-1/2n$, $\mb{v}[j]\in \{0,1\}$, 
and $\mb{v'}[j]<1+1/2n$. Therefore,  
\begin{align*}
S(\mb{v'})[k] & = \mb{v'}[k]+k/n > (1 - 1/2n) +k/n \\
& > (1+1/2n) + (k-1)/n > \mb{v'}[j] + j/n = S(\mb{v'})[j].
\end{align*}

Case 2: $0\le k<j\le n-1$. Note that $\mb{v}[j]=0$ because $low(\mb{v})=k$ is the largest index with $\mb{v}[k]=1$. Moreover, $\mb{v'}[k]>1-1/2n$ and $\mb{v'}[j]<1/2n$.  Therefore,
\begin{align*}
S(\mb{v'})[k] & = \mb{v'}[k]+k/n  > (1-1/2n)+k/n \ge 1-1/2n \\
& = 1/2n +  (n-1)/n > \mb{v'}[j] + j/n = S(\mb{v'})[j].    
\end{align*}
\end{appendixproof}


\begin{lemma}
\label{lemma:lowerrorexact} Let $\varepsilon > 0$ and let $\mb{e}_j$ denote the standard $n$-dimensional basis vector, with $\mb{e}_j[j]=1$ and $\mb{e}_j[i]=0$ for all $i \neq j$. Fix parameters $d, d', m, t$ for the $\MaxIdx$ algorithm so that $|\MaxIdx(\mb{x}; d,d',m,t) - \mb{e}_{maxidx(\mb{x})}| < 2^{-\alpha}$, for all $\mb{x} \in [\frac{1}{2}, \frac{3}{2})^n$. Then, for any binary vector $\mb{v} \in \{0,1\}^n$,  $$\big| \Low(\mb{v}; d, d', m, t) - low(\mb{v}) \big| < \frac{n(n-1)}{2}   (2^{-\alpha}),$$
where $\Low$, $low$, and $\MaxIdx$ are computed as described in \textup{\cite{cheon2019}}.
\end{lemma}
\begin{proof}
To ease notation let $\mb{b} =  \MaxIdx(\mb{v}; d, d', m, t)$. Assume the paramaters $d, d', m$ and $t$ have been chosen so that $$\big| \mb{b}[i] - \mb{e}_k[i] \big| < 2^{-\alpha}$$
for all $0 \leq i < n-1$, if $k = maxidx(\mb{v})$.
Then

\begin{align*}
\big| \texttt{Low}(\mb{v}; d, d', m, t) - low(\mb{v}) \big| &= \left| \sum_{i = 0}^{n-1} i\mb{b}[i] - \sum_{i = 0}^{n-1} i \mb{e}_j[i] \right| \\
&= \left| \sum_{i = 0}^{n-1} i(\mb{b}[i] - \mb{e}_j[i]) \right| \\
&\leq \sum_{i = 0}^{n-1} \big| i(\mb{b}[i] - \mb{e}_j[i]) \big| \\
&< 2^{-\alpha} \sum_{i = 0}^{n-1} i \\
&= \frac{n(n-1)}{2} \left(2^{-\alpha} \right).
\end{align*}
\end{proof}

\begin{lemma}
\label{lemma:lowerrorapprox}
Let $\alpha > 0$ and fix parameters $d, d', m, t$ for the $\MaxIdx$ algorithm so that $|\MaxIdx(\mb{x}; d,d',m,t) - \mb{e}_{maxidx(\mb{x})}| < 2^{-\alpha}$, for all $\mb{x} \in [\frac{1}{2}, \frac{3}{2})^n$. Further assume $\mb{v'} \in [0, 1]^n$ and $\mb{v} \in \{ 0, 1 \}^n$ are such that $|\mb{v'} - \mb{v}| < \frac{1}{2n}$. Then $$\big| \Low(\mb{v'}; d, d', m, t) - \Low(\mb{v}; d, d', m, t) \big| < n(n-1)   2^{-\alpha}$$
\end{lemma}
\begin{proof}
Recall that the $\Low$ algorithm presented in Algorithm \ref{alg:low} requires the input vectors to undergo two transformations, $S$ and $T_{\texttt{L}}$, before being fed into the $\MaxIdx$ algorithm. Let $\mb{x}' = T_{\texttt{L}}(S(\mb{v}'))$ and $\mb{x} = T_{\texttt{L}}(S(\mb{v}))$. The problem now reduces to bounding $$\big| \MaxIdx(\mb{x}'; d, d', m, t) - \MaxIdx(\mb{x}; d, d', m, t) \big|$$ where $\left| \mb{x}' - \mb{x} \right| < \frac{1}{4n} < \frac{1}{2n}$. By \textsc{Lemmas} \ref{lemma:low_maxidx_exact} and \ref{lemma:low_maxidx_approx}
\begin{align*}
maxidx(\mb{x}') & = maxidx(T_{\texttt{L}}(S(\mb{v}'))) = low(\mb{v})\\
& = maxidx(T_{\texttt{L}}(S(\mb{v}))) = maxidx(\mb{x}), 
\end{align*}
and so a triangle inequality obtained by adding and subtracting $\mb{e}_{maxidx(\mb{x'})}$ and $\mb{e}_{maxidx(\mb{x})}$ yields
\begin{align*}
    \big| \MaxIdx(\mb{x}'; d, d', m, t) &- \MaxIdx(\mb{x}; d, d', m, t) \big| \\ 
    < & 2^{-\alpha} + 0 + 2^{-\alpha} = 2^{-\alpha + 1}.
\end{align*}
Letting 
\begin{align*}
\mb{b} & =  \MaxIdx(\mb{x}; d, d', m, t)\\ 
\mb{b'} & =  \MaxIdx(\mb{x'}; d, d', m, t),    
\end{align*}
we have
\begin{align*}
    \big| \Low(\mb{v'}; d, d', m, t) - \Low(\mb{v}; d, d', m, t) \big| &= \left| \sum_{i = 0}^{n-1} i\mb{b}'[i] - \sum_{i = 0}^{n-1} i \mb{b}[i] \right| \\
    &=  \left| \sum_{i = 0}^{n-1} i(\mb{b}'[i] - \mb{b}[i]) \right|\\
    \leq & \sum_{i = 0}^{n-1} \big| i(\mb{b}'[i] - \mb{b}[i]) \big| < 2^{-\alpha + 1} \sum_{i = 0}^{n-1} i\\
    = & \frac{n(n-1)}{2} (2^{-\alpha + 1}) = n(n-1) 2^{-\alpha}.
\end{align*}
\end{proof}

\section{$\Low$ Parameters Proofs}
\label{s: Low parameters proofs}

\begin{prop}
\label{prop:c bound for approx}
Fix $n \geq 2$ and assume $\mb{v} \in \{ 0, 1 \}^n$ and $\mb{v'} \in [0, 1]^n$ satisfy $\left| \mb{v} - \mb{v'} \right| \leq \frac{\varepsilon}{2n}$, for some $0 \leq \varepsilon < 1$. Let $\mb{x'} = T_{\texttt{L}}(S(\mb{v'}))$ and denote the $(i+1)$-th smallest value of $\mb{x}'$ by $\mb{x}'_{(i)}$, so $\min \{\mb{x}'[i] \; | \; 0 \leq i \leq n-1\} =: \mb{x}'_{(0)} < \mb{x}'_{(1)} < \ldots < \mb{x}'_{(n-1)} := \max \{\mb{x}'[i] \; | \; 0 \leq i \leq n-1\}$. Let 
$$c = \mb{x}'_{(n-1)}/\mb{x}'_{(n-2)}$$ 
be the ratio of the largest coordinate value of $\mb{x}'$ over the second largest value. Then $$c \geq 1 + \frac{2-2\varepsilon}{6n-4+\varepsilon}.$$
\end{prop}

\begin{proof}
Suppose $low(\mb{v}) = k$. By \textsc{Lemma}~\ref{lemma:low_maxidx_approx}, $\mb{x}'[k]$ has the highest value in the vector. Define the two sets 
\begin{align*}M_1 &= \{j  \; | \; \mb{v}[j] = 1 \text{ and } j < k\} \\
&= \left\{j  \; | \; \mb{v}'[j] \geq 1 - \frac{\varepsilon}{2n} \text{ and } j < k \right\}
\end{align*} and $$M_0 = \{j  \; | \; \mb{v}[j] = 0\} = \left\{j  \; | \; \mb{v}'[j] \leq \frac{\varepsilon}{2n}\right\}.$$

There are two cases to consider, either $M_1$ is empty or not.

Case 1: If $M_1$ is empty, let $m = \max M_0$. Since $S$ is an increasing function with respect to index, and $T_{\texttt{L}}$ is strictly increasing we have that $T_{\texttt{L}}(S(\mb{v}'))[m] = \mb{x}'[m]$ is necessarily the second highest value. This is because  
\begin{align*}
\mb{x}'[m]  =  \frac{\mb{v}'[m] + \frac{m}{n} + 1}{2} \geq \frac{\frac{m}{n} + 1}{2} > \frac{\frac{\varepsilon}{2n} + \frac{j}{n} + 1}{2}  > \frac{\mb{v'}[j] + \frac{j}{n} + 1}{2} = \mb{x}'[j]   
\end{align*}
for $0 \leq j < m \leq n-1$. The middle inequality follows because $m,n \in \mathbb{Z}$ and $m > j \implies m-j \geq 1 \implies \frac{m}{n} \geq \frac{1}{n} + \frac{j}{n} > \frac{\varepsilon}{2n} + \frac{j}{n}$. Thus, 
\begin{align*}
    \frac{\mb{x}'_{(n-1)}}{\mb{x}'_{(n-2)}} = \frac{\mb{x}'[k]}{\mb{x}'[m]} &> \frac{1 - \frac{\varepsilon}{2n} +\frac{k}{n} + 1}{\frac{\varepsilon}{2n}+\frac{n-1}{n} + 1} \\
    &= 1 + \frac{2k+2-2\varepsilon}{4n-2+\varepsilon} \geq  1 + \frac{2-2\varepsilon}{4n-2+\varepsilon}.
\end{align*}

Case 2: If $M_1$ is not empty, let $m = \max M_1$. We first show that $\mb{x}'[i] > \mb{x}'[j]$ for all $i \in M_1, j \in M_0$; that is, all transformed approximate 1's are larger than all transformed approximate 0's. Let $j\in M_0$ and $i\in M_1$ be arbitrary. Then
\begin{align*}\mb{v}'[j] + \frac{j}{n} < \frac{\varepsilon}{2n} + \frac{j}{n} &\leq \frac{\varepsilon}{2n} + \frac{n-1}{n} \\
&\leq 1 - \frac{1}{2n} \leq 1 - \frac{\varepsilon}{2n} + \frac{i}{n} < \mb{v}'[i] + \frac{i}{n}
\end{align*}
for $i \geq 0$. Since $T_{\texttt{L}}$ is an increasing function, we have that $\mb{x}'[j] < \mb{x}'[i]$ as desired. Thus, the second largest coordinate of $\mb{x}'$ is $\mb{x}'[m] > (1 - \frac{\varepsilon}{2n} + \frac{m}{n} + 1)/2$, where $m = \max M_1$. Necessarily $m \leq k-1$, and so we have that 
\begin{align*}
\frac{\mb{x}'_{(n-1)}}{\mb{x}'_{(n-2)}} = \frac{\mb{x}'[k]}{\mb{x}'[m]} &> \frac{1 - \frac{\varepsilon}{2n} +\frac{k}{n} + 1}{1+\frac{k-1}{n} + 1} \\
&= 1 + \frac{2 - \varepsilon}{4n+2k-2} \geq 1 + \frac{2 - \varepsilon}{6n-4}.
\end{align*}
Thus, as $n$ and $\varepsilon$ were chosen from the beginning, we know that the ratio between the largest and second largest value for any ``approximate'' binary vector is bounded below by 
$\min\left\{1 + \frac{2-2\varepsilon}{4n - 2 + \varepsilon}, 1 + \frac{2-\varepsilon}{6n-4}\right\}$. However, both $ \frac{2-2\varepsilon}{4n - 2 + \varepsilon} \text{ and } \frac{2-\varepsilon}{6n-4}$ are greater than or equal to $\frac{2-2\varepsilon}{6n-4+\varepsilon}$, with the former being immediate for $n \geq 1$ and the latter is true as $$\frac{2 - \varepsilon}{6n-4} \geq \frac{2-2\varepsilon}{6n-4} \geq \frac{2-2\varepsilon}{6n-4+\varepsilon}, \text{ for }  0 \leq \varepsilon < 1.$$

And so, it follows that 
\[c \geq \min\left\{1 + \frac{2-2\varepsilon}{4n - 2 + \varepsilon}, 1 + \frac{2-\varepsilon}{6n-4}\right\} \geq 1 + \frac{2-2\varepsilon}{6n-4+\varepsilon}.\]
\end{proof}

\begin{lemma}
\label{lemma: alpha to delta}
Assume $\mb{v} \in \{ 0, 1 \}^n$ and $\mb{v'} \in [0, 1]^n$ satisfy $\left| \mb{v} - \mb{v'} \right| \leq \frac{\varepsilon}{2n}$, for some $0 \leq \varepsilon < 1$. Consider Theorem 5 in \cite{cheon2019}, which gives the parameters ($d, d', m, t$) for the $\MaxIdx$ function for a given desired accuracy of $2^{-\alpha}$. If $\alpha$ is chosen such that $$\alpha > \log(3) + 2\log(n) - \log(\delta) - 1$$ then $|\Low(\mb{v'}; d, d', m, t) - low(\mb{v})| < \delta$.
\end{lemma}

\begin{proof}
We note that from \textsc{Theorem}~\ref{theorem: lowerror} it suffices to have
\begin{align*}
& \frac{3}{2}n(n-1)2^{-\alpha} < \frac{3n^2}{2^{\alpha+1}} < \delta  \\
& \Leftrightarrow  2^{\alpha + 1} > \frac{3n^2}{\delta} \\
& \Leftrightarrow \alpha > \log(3) + 2\log(n) - \log(\delta) - 1
\end{align*}
\end{proof}

\begin{appendixproof}[\textbf{Proof of Theorem~\ref{theorem: low parms}}]

\textsc{Lemma}~\ref{lemma: alpha to delta} provides the corresponding $\alpha$ needed for a desired $\delta$-error in $\Low$. That is, if one desires error $\delta$ in $|\Low(\mb{v'}) - low(\mb{v})|$, then one may choose $\alpha > \log(3) + 2\log(n) - \log(\delta)-1$.

$c$'s lower bound has been found in \textsc{Proposition}~\ref{prop:c bound for approx}. In particular, by \textsc{Proposition}~\ref{prop:c bound for approx}, we know that $c \geq 1 + \frac{2-2\varepsilon}{6n-4+\varepsilon}$. But this is equivalent to saying that 
\begin{align*}
& -\log(\log(c)) \leq -\log(\log(1 + \frac{2-2\varepsilon}{6n-4+\varepsilon}))\\
\implies & \log(\alpha + 1 + \log(n)) -\log(\log(c)) \\
&\leq \log(\alpha + 1 + \log(n)) -\log(\log(1 + \frac{2-2\varepsilon}{6n-4+\varepsilon}))
\end{align*}
Therefore, if we choose 
$$t \geq \frac{ \log(\alpha + 1 + \log(n)) -\log(\log(1 + \frac{2-2\varepsilon}{6n-4+\varepsilon}))}{\log m}$$ we fulfill the former inequality in Theorem 5 (from \cite{cheon2019}).

At this point, both $\alpha$ and $t$ are determined. As $\min \{d, d' \}$ is dependent upon these choice of parameters, our choice of $\min \{d, d' \}$ is also determined.

\end{appendixproof}
\section{$\LowComp$ Correctness Proofs}
\label{s: LowComp correctness proofs}

\begin{appendixproof}[\textbf{Proof of \textsc{Theorem}~\ref{theorem: lowcomp error}}]

Recall \textsc{Lemma}~\ref{lem:approxlow} and the following chain of if and only if implications:
\begin{align*}
& lowcomp(l_\x, l_\y) = 1 \Leftrightarrow low(\x) = low(\y) \Leftrightarrow  \phi \geq |\texttt{L}_{\mb{x}'} - \texttt{L}_{\mb{y}'}| \\
& \Leftrightarrow \phi^2 \geq (\texttt{L}_{\mb{x}'} - \texttt{L}_{\mb{y}'})^2  \Leftrightarrow  T_{\texttt{C}}(\phi^2) \geq T_{\texttt{C}}(\texttt{L}_{\mb{x}'} - \texttt{L}_{\mb{y}'})^2
\end{align*}

There are two cases: $lowcomp(l_{\x}, l_{\y}) = 1$ or $lowcomp(l_{\x}, l_{\y}) = 0$.
Case 1: $lowcomp(l_{\x}, l_{\y}) = 1$. Then $T_{\texttt{C}}(\phi^2) \geq T_{\texttt{C}}(\texttt{L}_{\mb{x}'} - \texttt{L}_{\mb{y}'})^2$
which implies that 
\begin{align*}
    &|\Comp(T_{\texttt{C}}(\phi^2), T_{\texttt{C}}(\texttt{L}_{\mb{x}'} - \texttt{L}_{\mb{y}'})^2; d, d', m, t) \\
    &- comp(T_{\texttt{C}}(\phi^2), T_{\texttt{C}}(\texttt{L}_{\mb{x}'} - \texttt{L}_{\mb{y}'})^2)| < \eta \\
    \implies &|\Comp(T_{\texttt{C}}(\phi^2), T_{\texttt{C}}(\texttt{L}_{\mb{x}'} - \texttt{L}_{\mb{y}'})^2; d, d', m, t) - 1| < \eta \\
    \implies &1 - \eta < \Comp(T_{\texttt{C}}(\phi^2), T_{\texttt{C}}(\texttt{L}_{\mb{x}'} - \texttt{L}_{\mb{y}'})^2; d, d', m, t) < 1\\
    \implies &1 - \eta < \LowComp\texttt{L}_{\mb{x}'}, \texttt{L}_{\mb{y}'}; d, d', m, t) < 1.
\end{align*}

Case 2: $lowcomp(l_{\x}, l_{\y}) = 0$. Then $T_{\texttt{C}}(\phi^2) < T_{\texttt{C}}(\texttt{L}_{\mb{x}'} - \texttt{L}_{\mb{y}'})^2$
which implies that 
\begin{align*}
    &|\Comp(T_{\texttt{C}}(\phi^2), T_{\texttt{C}}(\texttt{L}_{\mb{x}'} - \texttt{L}_{\mb{y}'})^2; d, d', m, t) \\
    & - comp(T_{\texttt{C}}(\phi^2), T_{\texttt{C}}(\texttt{L}_{\mb{x}'} - \texttt{L}_{\mb{y}'})^2)| < \eta \\
    \implies &|\Comp(T_{\texttt{C}}(\phi^2), T_{\texttt{C}}(\texttt{L}_{\mb{x}'} - \texttt{L}_{\mb{y}'})^2; d, d', m, t) - 0| < \eta \\
    \implies &0 < \Comp(T_{\texttt{C}}(\phi^2), T_{\texttt{C}}(\texttt{L}_{\mb{x}'} - \texttt{L}_{\mb{y}'})^2; d, d', m, t) < \eta\\
    \implies &0 < \LowComp(\texttt{L}_{\mb{x}'}, \texttt{L}_{\mb{y}'}; d, d', m, t) < \eta.
\end{align*}
\end{appendixproof}
\section{$\LowComp$ Parameters Proofs}
\label{s: LowComp parameters proofs}

\begin{prop}
\label{prop: geometric mean}
Consider positive $a < b < c$. Then the value of $b$ which optimizes the problem $$\max_{b \in (a, c)} \left( \min\left\{\frac{b}{a}, \frac{c}{b} \right\} \right)$$ is $b = \sqrt{ac}$. In other words, $b$ is equal to the geometric mean of the endpoints. Furthermore, for this choice of $b$, we have that $$\max_{b \in (a, c)} \left( \min\left\{\frac{b}{a}, \frac{c}{b} \right\} \right) = \sqrt{\frac{c}{a}}.$$
\end{prop}
\begin{proof}

Suppose $b = \sqrt{ac}$. Then $\frac{b}{a} = \frac{c}{b} = \sqrt{\frac{c}{a}}$. If $b > \sqrt{ac}$, then $\frac{b}{a} > \frac{c}{b}$, implying that the minimum of the two ratios is $\frac{c}{b}$. But then $\frac{c}{b} < \frac{c}{\sqrt{ac}} = \sqrt{\frac{c}{a}}$. Similarly, if $b < \sqrt{ac}$, then $\frac{b}{a} < \frac{c}{b}$, so that the minimum of the two ratios is now $\frac{b}{a}$. It follows that $\frac{b}{a} < \frac{\sqrt{ac}}{a} = \sqrt{\frac{c}{a}}.$ Thus the value of $b$ over $(a,c)$ which maximizes the $\min\left\{\frac{b}{a},\frac{c}{b}\right\}$ is $b = \sqrt{ac}$ and the maximum value is $\sqrt{\frac{c}{a}}$.

\end{proof}

\begin{appendixproof}[\textbf{Proof of \textsc{Corollary}~\ref{corollary: lowcomp c bound}}]
The value of $T_{\texttt{C}}(\phi^2)$ comes immediately from applying the preceding \textsc{Proposition} \ref{prop: geometric mean}. This proposition also establishes $c$'s lower bound, as
\begin{align*}
c & = \frac{\max \left\{ T_{\texttt{C}}(\phi^2), T_{\texttt{C}}((\texttt{L}_{\x'} - \texttt{L}_{\x'})^2) \right\} }{\min \left\{ T_{\texttt{C}}(\phi^2), T_{\texttt{C}}((\texttt{L}_{\x'} - \texttt{L}_{\x'})^2) \right\} } \\
 &> \min \left\{ \frac{T_{\texttt{C}}(\phi^2)}{T_{\texttt{C}}((2\delta)^2)} , \frac{T_{\texttt{C}}((1-2\delta)^2)}{T_{\texttt{C}}(\phi^2)} \right\} = \sqrt{\frac{T_{\texttt{C}}((1-2\delta)^2)}{T_{\texttt{C}}((2\delta)^2)}} \\
& = \sqrt{\frac{\frac{1}{2} + (\frac{1-2\delta}{n})^2}{\frac{1}{2} + (\frac{2\delta}{n})^2}}  = \sqrt{\frac{n^2 + 2(1-2\delta)^2}{n^2 + 2(2\delta)^2}}.
\end{align*}
\end{appendixproof}

\begin{lemma}
\label{lemma: alpha to eta}
Consider \textsc{Theorem}~\ref{Comp parms}, which gives the parameters ($d, d', m, t$) for the $\Comp$ function for a given desired accuracy of $2^{-\alpha}$. If $\alpha$ is chosen such that $\alpha > -\log(\eta)$ then $\Comp$ (and by extension, $\LowComp$) has $\eta$-error. That is, $$|\LowComp(\textup{\texttt{L}}_{\mb{x}'}, \textup{\texttt{L}}_{\mb{y}'}, \phi; d, d', m, t) - lowcomp(l_{\x}, l_{\y})| < \eta$$
\end{lemma}
\begin{proof}
If $\alpha > -\log(\eta)$, then $2^{-\alpha} < \eta$. The result follows from \textsc{Theorem} \ref{theorem: lowcomp error}. 
\end{proof}

\begin{appendixproof}[\textbf{Proof of \textsc{Theorem}~\ref{theorem: lowcomp parms}}]
For 
\[\phi = n\sqrt{\sqrt{\left(\frac{1}{2} + (\frac{2\delta}{n})^2 \right) \left(\frac{1}{2} + (\frac{1-2\delta}{n})^2 \right)} - \frac{1}{2}},\] 
one may calculate 
\[T_{\texttt{C}}(\phi^2) = \sqrt{\left(\frac{1}{2} + (\frac{2\delta}{n})^2 \right) \left(\frac{1}{2} + (\frac{1-2\delta}{n})^2 \right)}.\] 
By \textsc{Corollary} \ref{corollary: lowcomp c bound}, $c$ is bounded below by $\sqrt{\frac{n^2 + 2(1-2\delta)^2}{n^2 + 2(2\delta)^2}}$. But this is equivalent to saying that 
\begin{align*}
& -\log(\log(c)) \leq -\log(\log(\sqrt{\frac{n^2 + 2(1-2\delta)^2}{n^2 + 2(2\delta)^2}})) \\
\implies &  \log(\alpha + 2) -\log(\log(c)) \\
& \leq \log(\alpha + 2) -\log(\log(\sqrt{\frac{n^2 + 2(1-2\delta)^2}{n^2 + 2(2\delta)^2}})).
\end{align*}
Therefore, if we choose 
\begin{align*}t \geq \frac{1}{\log m}\Bigg[ &\log(\alpha + 1 + \log(n)) \\ &-\log(\log(\sqrt{\frac{n^2 + 2(1-2\delta)^2}{n^2 + 2(2\delta)^2}})) \Bigg],\end{align*} we fulfill the former inequality in Theorem 4.

Furthermore, \textsc{Lemma}~\ref{lemma: alpha to eta} determines what $\alpha$ must be to achieve the desired $\eta$-error in $\LowComp$. At this point, with a fixed choice of $m$, Theorem 4 from \cite{cheon2019} establishes the remaining parameters $d, d',$ and $t$.
\end{appendixproof}

\end{document}